\documentclass[prepint]{elsarticle}
\usepackage{easy-todo}
\usepackage[english]{babel}

\usepackage[T1]{fontenc}
\usepackage[utf8]{inputenc}
\usepackage[table]{xcolor}
\usepackage{afterpage}
\usepackage{tabularx}
\usepackage{geometry}
\usepackage{doi}
\usepackage{natbib}
\usepackage[english]{babel}
\usepackage{csquotes}
\usepackage{fleqn}
\usepackage{amssymb}
\usepackage{amsmath}
\usepackage{relsize}
\usepackage{graphicx}
\usepackage{bm}
\usepackage{bussproofs}
\usepackage{mathtools}
\usepackage{amsthm}
\usepackage{stmaryrd}
\usepackage{xspace}
\usepackage{comment}
\usepackage{cleveref}

\usepackage[inline]{enumitem}

\newcommand{\pill}{\ensuremath{\pi\mathsf{ILL}}\xspace}

\newcommand{\pcll}{\ensuremath{\pi\mathsf{CLL}}\xspace}
\newcommand{\cp}{\ensuremath{\mathsf{CP}}\xspace}
\newcommand{\pcllcp}{\ensuremath{\pcll^\star}\xspace}
\newcommand{\pcllcpp}{\ensuremath{\pcllcp_{+}}\xspace}
\newcommand{\pcllcppp}{\ensuremath{\pcllcp_{+\mkern-8mu +}}\xspace}
\newcommand{\ull}{\ensuremath{\mathsf{ULL}}\xspace}
\newcommand{\pull}{\ensuremath{\pi\mathsf{ULL}}\xspace}
\newcommand{\pullm}{\ensuremath{\pull_\curvearrow}\xspace}
\newcommand{\pullp}{\ensuremath{\pull_{+}}\xspace}
\newcommand{\pullpp}{\ensuremath{\pull_{+\mkern-8mu +}}\xspace}
\newcommand{\vpull}{\ensuremath{\vdash_{\mkern-10mu \scriptscriptstyle \curvearrow \mkern-5mu}}\xspace}
\newcommand{\vdill}{\ensuremath{\vdash_{\mathcal{I}}}\xspace}
\newcommand{\vdcll}{\ensuremath{\vdash_{\mathcal{C}}}\xspace}
\newcommand{\U}{\ensuremath{\mathcal{U}}\xspace}
\newcommand{\C}{\ensuremath{\mathcal{C}}\xspace}
\newcommand{\I}{\ensuremath{\mathcal{I}}\xspace}

\newcommand{\infrule}[1]{\ensuremath{\text{\textsc{#1}}}}

\newcommand{\id}{\ensuremath{\infrule{Id}}\xspace}

\let\onu\nu
\renewcommand{\nu}[1]{\ensuremath{\textstyle {(\onu #1)}}}
\let\oparallel\|
\renewcommand{\|}{\ensuremath{\mathbin{|}}}
\let\ooplus\oplus
\renewcommand{\oplus}{{\ooplus}}
\newcommand{\fwd}[2]{\ensuremath{[#1 \mathbin{\leftrightarrow} #2]}}
\newcommand{\0}{\ensuremath{\bm{0}}}
\newcommand{\send}[2]{\ensuremath{#1 \langle #2 \rangle}}
\newcommand{\recv}[2]{\ensuremath{#1(#2)}}
\newcommand{\serv}[2]{\ensuremath{{!}#1(#2)}}
\newcommand{\subst}[1]{\ensuremath{\{#1\}}}

\newcommand{\rott}[1]{\mathpalette\rot{#1}}
\newcommand{\rot}[2]{\rotatebox[origin=c]{180}{$#1{#2}$}}

\let\obot\bot
\renewcommand{\bot}{\ensuremath{\obot}\xspace}
\let\ocdot\cdot
\renewcommand{\cdot}{\ensuremath{\ocdot}\xspace}
\newcommand{\tensor}{\ensuremath{\mathbin{\otimes}}\xspace}
\newcommand{\lolli}{\ensuremath{\mathbin{\multimap}}\xspace}
\newcommand{\parr}{\ensuremath{\mathbin{\rott{\&}}}\xspace}
\newcommand{\1}{\ensuremath{\bm{1}}}
\newcommand{\sbot}{\ensuremath{\mathsmaller{\bot}}}
\newcommand{\bang}{\ensuremath{{!}}\xspace}
\newcommand{\whynot}{\ensuremath{{?}}\xspace}
\newcommand{\dual}[1]{\ensuremath{#1^\sbot}}

\newcommand{\mright}{\ensuremath{{\curvearrowright}}\xspace}
\newcommand{\mleft}{\ensuremath{{\curvearrowleft}}\xspace}
\newcommand{\curvearrow}{\ensuremath{{\scurvearrow}}\xspace}
\newcommand{\scurvearrow}{
    \mathchoice
    {{\curvearrowleft}\mkern-15.7mu{\curvearrowright}}
    {{\curvearrowleft}\mkern-15.7mu{\curvearrowright}}
    {{\curvearrowleft}\mkern-15.4mu{\curvearrowright}}
    {{\curvearrowleft}\mkern-14.8mu{\curvearrowright}}
}

\newcommand{\fn}{\ensuremath{\text{fn}}}
\newcommand{\bn}{\ensuremath{\text{bn}}}
\newcommand{\red}{\ensuremath{\mathrel{\rightarrow}}}
\DeclareMathOperator{\nred}{\not\rightarrow}
\newcommand{\bnf}{\ensuremath{\textstyle \mathbin{\mathlarger{\mathlarger{|}}}}}

\newcommand{\rdegree}{$r$-degree\xspace}

\newcommand{\tend}{\ensuremath{\mathsf{end}}\xspace}

\newtheorem{theorem}{Theorem}[section]

\newtheorem{corollary}[theorem]{Corollary}
\newtheorem{definition}{Definition}[section]

\newcommand{\etal}{\emph{et al.}\xspace}

\newcommand{\secref}[1]{\S\,\labelcref{#1}}

\usepackage{hyperref}
\hypersetup{
    hypertexnames=true,
}

\def\@nonitbrack#1#2{{\normalfont{#1[#2]}}}

\NewDocumentCommand\ruleLabel{s O{} m}{%
    \@nonitbrack{#2}{\textsc{#3}\IfBooleanT{#1}{-\ensuremath{\ast}}}%
}

\def\inferLabel#1#2{\mbox{\scriptsize\normalfont{\ruleLabel{#1}#2}}}

\RequirePackage{bussproofs}
\RequirePackage{xparse}
\NewDocumentCommand\bussRL{O{}}{%
    \IfValueT{#1}{%
        \RightLabel{\scriptsize #1}
    }
}
\NewDocumentCommand\bussAx{O{} m}{\AxiomC{}\bussRL[#1]\UnaryInfC{\ensuremath{#2}}}
\def\bussAssume#1{\AxiomC{\ensuremath{#1}}}
\NewDocumentCommand\bussUn{O{} m}{\bussRL[#1]\UnaryInfC{\ensuremath{#2}}}
\NewDocumentCommand\bussBin{O{} m}{\bussRL[#1]\BinaryInfC{\ensuremath{#2}}}
\NewDocumentCommand\bussTern{O{} m}{\bussRL[#1]\TrinaryInfC{\ensuremath{#2}}}
\NewDocumentCommand\bussQuat{O{} m}{\bussRL[#1]\QuaternaryInfC{\ensuremath{#2}}}
\def\bussDisplay{\DisplayProof}
\NewDocumentEnvironment{bussproof}{o O{} D(){b}}{%
    \begin{array}[#3]{@{}l@{}}
        \IfValueT{#1}{%
            \inferLabel{#1}{#2}
            \\
        }%
        \bottomAlignProof
}{%
        \bussDisplay%
    \end{array}
}

\def\ih#1{IH\textsubscript{#1}}
\usepackage{mathpartir}
\newcommand{\plscheck}[1]{{#1}}

\def\beginbas{}
\def\endbas{}
\let\beginjp\beginbas
\let\endjp\endbas

\nopreprintlinetrue

\begin{document}

\begin{frontmatter}

\title{Comparing Session Type Systems derived from Linear Logic\tnoteref{t1}}
\tnotetext[t1]{
    Work partially supported by the Dutch Research Council (NWO) under the VIDI Project No.\ 016.Vidi.189.046 (Unifying Correctness for Communicating Software).
}

\author[1,2]{Bas van den Heuvel}
\author[3]{Jorge A. Pérez}

\affiliation[1]{
    organization={HKA Karlsruhe},
    city={Karlsruhe},
    country={Germany}
}
\affiliation[2]{
    organization={University of Freiburg},
    city={Freiburg},
    country={Germany}
}
\affiliation[3]{
    organization={University of Groningen},
    city={Groningen},
    country={The Netherlands}
}


\begin{abstract}
    \emph{Session types} are a typed approach to message-passing concurrency, where types describe sequences of intended exchanges over channels.
   Session type systems have been given strong logical foundations via Curry-Howard correspondences with \emph{linear logic}, a resource-aware logic that naturally captures structured interactions.
    These logical foundations provide an elegant framework to specify and (statically) verify message-passing processes.

    In this paper, we rigorously compare different type systems for concurrency derived from the Curry-Howard correspondence between linear logic and session types.
   We address the main divide between these type systems: the classical and intuitionistic presentations of linear logic.
   Over the years, these presentations have given rise to separate research strands on logical foundations for concurrency; the differences between their derived type systems have only been addressed informally.

    To formally assess these differences, we develop
   \pull, a session type system that encompasses type systems derived from classical and intuitionistic interpretations of linear logic.
   Based on a fragment of Girard's Logic of Unity, \pull provides a basic reference framework:
    we compare existing session type systems by characterizing fragments of \pull that coincide with classical and intuitionistic formulations.
    We analyze the significance of our characterizations by considering  the \emph{locality} principle (enforced by intuitionistic interpretations but not by classical ones) and forms of \emph{process composition} induced by the interpretations.
\end{abstract}

\begin{keyword}
Concurrency, linear logic, $\pi$-calculus, session types.
\end{keyword}

\end{frontmatter}

\section{Introduction}

Establishing the correctness of \emph{message-passing} programs is a central but challenging problem.
Within formal methods for concurrency, \emph{session types}~\cite{conf/concur/Honda93,conf/parle/TakeuchiHK94,conf/esop/HondaVK98} are now a consolidated approach to statically verifying safety and liveness properties of communicating programs.
Session types specify communication over channels as sequences of exchanges.
This way, e.g., the session type $!\mathsf{int}.?\mathsf{bool}.\tend$ describes a channel's intended protocol:
send an integer, receive a boolean, and close the channel.
Due to its simplicity and expressiveness, the $\pi$-calculus~\cite{journal/ic/MilnerPW92,book/SangiorgiW03}---the paradigmatic model of concurrency and interaction---is a widely used specification language for developing session types and establishing their basic theory.
In this paper, we are interested in developing further the \emph{logical foundations} of session type systems for the $\pi$-calculus.

\paragraph{Context: Linear Logic and Session Types}

In a line of work developed by Caires, Pfenning, Wadler, and several others, the theory of session types has been given firm logical foundations in the form of Curry-Howard-style correspondences.
Caires and Pfenning~\cite{conf/concur/CairesP10} discovered a correspondence between session types for the $\pi$-calculus and Girard's \emph{linear logic}~\cite{journal/tcs/Girard87}:
\begin{center}
    \begin{tabular}{rcl}
          session types
        & $\Leftrightarrow$
        & logical propositions
        \\
          $\pi$-calculus processes
        & $\Leftrightarrow$
        & proofs
        \\
          process communication
        & $\Leftrightarrow$
        & cut elimination
    \end{tabular}
\end{center}
Based on this logical bridge, the resulting type systems simultaneously ensure important properties for processes, such as \emph{session fidelity} (processes respect session types), \emph{communication safety} (absence of communication errors), and \emph{progress/deadlock freedom} (processes never reach stuck states).

There are two presentations of linear logic, \emph{classical} and \emph{intuitionistic}, and session type systems derived from linear logic have inherited this dichotomy.
Under Curry-Howard interpretations, typing judgments and typing rules correspond to the logical sequents and inference rules of the underlying linear logic.
In both classical and intuitionistic cases, judgments assign independent session protocols to the channels of a process.
Because  these judgments differ between the presentations of linear logic, there are consequences for their respective interpretations as type systems:
\begin{itemize}
	\item
Caires and Pfenning's correspondence~\cite{conf/concur/CairesP10} uses an \emph{intuitionistic} linear logic where judgments for process are two-sided, with zero or many channels on the left and exactly one channel on the right.
Such judgments have a convenient rely-guarantee reading: the process \emph{relies} on the behaviors described by the channels on the left to \emph{guarantee} the behavior of the one channel on the right.
In this interpretation, each logical connective thus requires two rules: the right rule specifies how to \emph{offer} a behavior, while the left rule specifies how to \emph{use} a behavior.
\item  Wadler's correspondence~\cite{conf/icfp/Wadler12} uses  \emph{classical} linear logic, where judgments are single-sided (all channels appear on the right) and lack a rely-guarantee reading.
In this interpretation, there is only a single rule per logical connective, which makes such type systems direct and economical.
\end{itemize}

Although the differences and relationship between intuitionistic and classical linear logic are relevant and well-known (see, e.g.,~\cite{journal/jlc/Schellinx91,report/ChangCP03,conf/lics/L18}), the differences between their derived session type systems have only been addressed informally.
In particular, Caires \etal~\cite{journal/mscs/CairesPT16} observe that, unlike classical interpretations, an intuitionistic interpretation guarantees \beginbas \emph{locality} for shared channels, i.e., the principle by which received channels cannot become servers (see \Cref{ssec:locality})\endbas.
In the meantime, both interpretations have been extended in multiple, exciting directions (see, e.g.,~\cite{conf/esop/CairesPPT13,chapter/AtkeyLM16,conf/csl/DeYoungCPT12,conf/tgc/ToninhoCP14,conf/icfp/LindleyM16,conf/icfp/BalzerP17,conf/esop/CairesP17,conf/fossacs/DardhaG18,journal/toplas/ToninhoY18,conf/concur/CairesPPT19,conf/popl/KokkeMP19,conf/icfp/QianKB21,conf/oopsla/FruminDHP22,conf/places/HorneP23,conf/aplas/vdHeuvelPNP23}).
This emergence of \emph{families} of type systems (one classical, one intuitionistic) has deepened the original dichotomy and somewhat obscured our understanding of  logical foundations of concurrency as a whole.
This state of affairs calls for a  formal comparison between the session type systems derived from classical and intuitionistic linear logics that goes beyond their superficial differences. To us, this seems an indispensable step in consolidating the logical foundations of message-passing concurrency.


\paragraph{Goal of the Paper}

In this paper, we aim at formally comparing the session type systems derived from interpretations of intuitionistic and classical linear logic.
We are concerned with two families of type systems for the $\pi$-calculus, which naturally induce two classes of typable processes, namely those typable in intuitionistic and classical interpretations, respectively.
These two classes are known to largely overlap with each other;
informal observations suggest that some processes are typable in a classical setting but not in an intuitionistic setting.
We are interested in determining precisely \emph{how} these classes of processes relate to each other.

The key step in our approach is to define a basic framework of reference in which both type systems, intuitionistic and classical, can be objectively compared.
Girard's Logic of Unity (LU)~\cite{journal/apal/Girard93} has a similar goal: to objectively compare classical logic, intuitionistic logic, and (classical and intuitionistic) linear logic in a single system, abstracting away from syntactical differences.
It is then only natural to use LU as a reference for our comparison.
We develop \pull, a session type system for the $\pi$-calculus that subsumes classical and intuitionistic interpretations of linear logic as session types.
{Based on a fragment of LU that suits our purposes (dubbed \ull), the type system \pull
 allows us to define a class of typable processes that  contains processes typable in both the intuitionistic and classical type systems.}

In our approach, the type system \pull provides an effective framework and a yardstick for a formal comparison.
Because its typing rules encompass both classical and intuitionistic formulations, we can readily characterize their two (sub-)classes of typable processes by considering the different typing rules that apply in each case.
Analyzing these different portions of rules allows us to then characterize the differences (in expressiveness/typability) between the intuitionistic and classical type systems.
Moreover, these characterizations allow \beginbas us \endbas to determine how extensions to the type systems might be supported by the intuitionistic and classical interpretations.

\paragraph{Contributions and Outline}

Summing up, this paper makes the following contributions:
\begin{enumerate}[noitemsep]
    \item
    The session type system \pull, which is derived from a concurrent interpretation of \ull (a linear fragment of LU), together with its correctness results (\Cref{thm:subjectcongruence,thm:subjectreduction,thm:progress,thm:df}, \secref{sec:ull});

    \item
        A formal comparison between (i)~the classes of processes typable under interpretations of intuitionistic and classical linear logic and (ii)~those typable in \pull~(\secref{sec:comparison}).
        We prove that
        (1)~\pull \emph{precisely captures} the class of processes typable under the classical interpretation (\Cref{thm:cllull}), and that
        (2)~the class of processes typable under the intuitionistic interpretation is \emph{strictly included} in \pull (\Cref{thm:illsubsetull}).
        Together, these two results confirm the informal observation that the two interpretations induce different classes of processes.

    \item
        An analysis of the significance of our characterizations in terms of two different aspects: the \emph{locality} principle enforced by intuitionistic interpretations, and \emph{process composition} usually limited in type systems derived from linear logic in comparison to session type systems not based on logical foundations  (\secref{sec:discussion}).
        This analysis corroborates and \beginjp gives a formal footing to the informal \endjp     observation made in~\cite{journal/mscs/CairesPT16} concerning locality, and suitably extends it to show that intuitionistic interpretations
    cannot type empty sends on previously received channels.

\end{enumerate}
Finally, we draw some conclusions in~\Cref{sec:conclusion}.

\paragraph{Related Work}

Our work can be seen as an expressiveness study on the classes of processes induced by different type systems.
 Concretely, our results can be seen as providing formal evidence that session type systems derived from classical linear logic are \emph{more expressive} than those based on intuitionistic linear logic, based on the fact that the former are \emph{more permissive} than the latter.

The study of the relative expressiveness of process calculi has a long, fruitful history (see, e.g.~\cite{journal/dc/Gorla10,conf/express/Peters19,thesis/Perez10}).
The aim is to compare two different process languages by defining an \emph{encoding} (a translation up to some correctness criteria) between them or by proving its non-existence.
Even though our aims are similar in spirit to expressiveness studies, at a technical level there are substantial differences, because our focus is on assessing the influence that different type systems have on the same process language---as such, our comparisons do not rely on encodings.

Although most studies on relative expressiveness have addressed \emph{untyped} process languages, some works have studied  the influence of (behavioral) types on the expressiveness of process languages---see, e.g.,~\cite{conf/concur/DemangeonH11,conf/express/GayGR14,journal/ic/DardhaGS17,journal/ic/KouzapasPY19,conf/forte/Perez16}.
Salient examples are the works \cite{journal/jlamp/DardhaP22,DBLP:conf/ppdp/Paulus0N23}, which compare type systems that enforce the deadlock-freedom and termination properties, respectively. 
In particular, a main result in~\cite{journal/jlamp/DardhaP22} is that (classical) linear logic induces a strict subclass of deadlock-free processes with respect to the class of typable processes in non-logically motivated type systems.
Unlike our work, the focus in~\cite{journal/jlamp/DardhaP22} is on different process languages, each with a different type system.
This requires the definition of encodings, on processes but also on types, that abstract away from the syntactic differences between the classes of processes induced by each typed framework.

\medskip
This paper is an extended, revised version of the workshop paper~\cite{conf/places/vdHeuvelP20}.
We present several improvements and developments with respect to that work.
First, here we have significantly improved the correctness results for \pull (\Cref{thm:subjectcongruence,thm:subjectreduction,thm:progress,thm:df})
and considered branching/selection types (not studied in~\cite{conf/places/vdHeuvelP20}).
Also, we have refined the presentation of \pull in~\cite{conf/places/vdHeuvelP20} with a more explicit treatment of duality, presented in~\Cref{ssec:duality}. Moreover, we have broadened the analysis of our comparison results in \Cref{sec:discussion} by  considering extensions to the type system with more expressive forms of parallel composition and restriction.

\section{A Session Type System Based on LU}
\label{sec:ull}

Girard~\cite{journal/apal/Girard93} developed Logic of Unity (LU) to study and compare classical, intuitionistic, and linear logic, without having to ``change the rules of the game by, e.g., passing from one style of sequent to the other.''
The idea is simple: there is one form of sequent with an abundance of inference rules.
The several logics subsumed by LU are then characterized by (possibly overlapping) subsets of those rules.

Clearly, we find ourselves in a similar situation:
we want to compare intuitionistic and classical linear logic as session type systems, by abstracting away from typing judgments and rules of different forms.
To this end, in this section we introduce United Linear Logic (\ull), a logic based on the linear fragment of LU, and present the Curry-Howard interpretation of \ull as a session type system for the $\pi$-calculus, dubbed \pull, following~\cite{conf/concur/CairesP10,conf/icfp/Wadler12} (\secref{ssec:proccal}).
As we will see in \Cref{sec:comparison}, we can then characterize the session type interpretations of intuitionistic and classical linear logic as subsets of the typing rules of \pull.
We also present correctness properties common to (logically motivated) session type systems~(\secref{ssec:procprop}).
Finally, we discuss an alternative presentation of \pull that has a more explicit account of duality (\secref{ssec:duality}).

\subsection{The Process Calculus and Type System}
\label{ssec:proccal}

\paragraph{Propositions/Types}

A session type represents a sequence of exchanges that should be performed along a channel.
Propositions in \ull---interpreted as session types in \pull---are defined as follows:

\begin{definition}\label{def:propositions}
    \emph{\ull propositions/\pull types} are generated by the following grammar:
    \begin{align*}
        A, B ::= \1 \bnf \bot \bnf A \tensor B \bnf A \lolli B \bnf \oplus\{i:A\}_{i \in I} \bnf \&\{i:A\}_{i \in I} \bnf \bang A \bnf \whynot A
    \end{align*}
\end{definition}

\noindent
\Cref{tbl:propsassessiontypes} gives the intuitive reading of the interpretation of propositions as session types.
Some details are worth noting.
\begin{itemize}
    \item
        First, receiving in \pull is defined using the intuitionistic connective \lolli;
        we will see later that \lolli can be used to define the classical connective \parr---also interpreted as receiving---which is not supported by intuitionistic linear logic.
    \item
        Second, \pull supports labeled $n$-ary branching constructs (following, e.g., Caires and Pérez in~\cite{conf/esop/CairesP17}).
        In linear logic, $\oplus$ and $\&$ are binary connectives~\cite{journal/tcs/Girard87,DBLP:conf/tapsoft/GirardL87,journal/apal/Girard93}.
        In~\cite{conf/concur/CairesP10,conf/icfp/Wadler12}, Caires, Pfenning, and Wadler originally interpret them as a choice between a left and a right option.
        However, in session-typed $\pi$-calculi (see, e.g.,~\cite{conf/esop/HondaVK98}) $n$-ary choice is standard and more convenient.
    \item
        \beginjp
        Third, the syntax of formulas considered by Girard~\cite{journal/apal/Girard93} subsumes
            our syntax for \ull by including  universal and existential quantifiers and the units $\top$ and $\0$. Our syntax does not include quantifiers and considers both $\1$ and $\bot$: while the interpretation of intuitionistic linear logic by Caires and Pfenning includes $\1$ (but not $\top$, $\0$, or $\bot$), the interpretation of classical linear logic by Wadler considers quantifiers and the units $\1$ and $\bot$ (but not $\top$ or $\0$).
        \endjp

\end{itemize}

{
    \rowcolors{1}{black!3}{black!5}
    \begin{table}[t]
        \begin{center}
            \begin{tabular}{ll}
                  $\1$ and $\bot$
                & Close the channel
                \\
                  $A \tensor B$
                & Send a channel of type $A$ and continue as $B$
                \\
                  $A \lolli B$
                & Receive a channel of type $A$ and continue as $B$
                \\
                  $\oplus\{i:A_i\}_{i \in I}$
                & Send a label $i \in I$ and continue as $A_i$
                \\
                  $\&\{i:A_i\}_{i \in I}$
                & Receive a label $i \in I$ and continue as $A_i$
                \\
                  $\bang A$
                & Repeatedly provide a service of type $A$
                \\
                  $\whynot A$
                & Connect to a service of type $A$
            \end{tabular}
        \end{center}
        \caption{Interpretation of \ull propositions as session types.}
        \label{tbl:propsassessiontypes}
    \end{table}
}

\paragraph{Duality}

The duality of \ull propositions is defined as follows:

\begin{definition}\label{def:duality}
    \emph{Duality} is a binary relation on propositions/types, denoted $\dual{A}$, defined as follows:
    \begin{align*}
        \dual{\1}
        &:= \bot
        & \dual{(A \tensor B)}
        &:= A \lolli \dual{B}
        & \dual{(\oplus\{i:A_i\}_{i \in I})}
        &:= \&\{i:\dual{A_i}\}_{i \in I}
        & \dual{(\bang A)}
        &:= \whynot \dual{A} \\
        \dual{\bot}
        &:= \1
        & \dual{(A \lolli B)}
        &:= A \tensor \dual{B}
        & \dual{(\&\{i:A_i\}_{i \in I})}
        &:= \oplus\{i:\dual{A_i}\}_{i \in I}
        & \dual{(\whynot A)}
        &:= \bang \dual{A}
    \end{align*}
\end{definition}

\noindent
Duality in \pull reflects the intended reciprocity of protocols between two parties:
when a process on one side of a channel sends, the process on the opposite side must receive, and vice versa.
It is easy to see that duality is an involution:
$\dual{(\dual{A})} = A$.

As mentioned before, we can use duality to define \parr in terms of \lolli: $A \parr B := \dual{A} \lolli B$;
as we will see, the rules for \parr of classical linear logic can be recovered from the rules for \lolli in \ull using duality.
Duality also shows that \tensor and \parr are De Morgan-style duals in classical linear logic:
\begin{align*}
    \dual{(A \tensor B)}
    &= A \lolli \dual{B} = \dual{(\dual{A})} \lolli \dual{B} = \dual{A} \parr \dual{B} \\
    \dual{(A \parr B)}
    &= \dual{(\dual{A} \lolli B)} = \dual{A} \tensor \dual{B}
\end{align*}

\paragraph{Processes}

The discipline of binary session types deals with concurrent \emph{processes} that communicate through point-to-point \emph{channels}.
The $\pi$-calculus~\cite{journal/ic/MilnerPW92,book/SangiorgiW03} offers a rigorous yet expressive framework for defining this discipline and establishing its fundamental properties.
\pull is a type system for $\pi$-calculus processes defined as follows:

\begin{definition}\label{def:processterms}
    \emph{Process terms} are generated by the following grammar:
    \begin{align*}
        P, Q ::=
        &~
        \0 \bnf \nu{x}P \bnf P \| Q \bnf \send{x}{y}.P \bnf \recv{x}{y}.P \bnf x \triangleleft \ell.P \bnf x \triangleright \{i:P\}_{i \in I} \\
        \bnf
        &~
        \serv{x}{y}.P \bnf \fwd{x}{y} \bnf \send{x}{}.P \bnf \recv{x}{}.P
    \end{align*}
\end{definition}

\noindent
Process constructs for inaction $\0$, channel restriction $\nu{x}P$, and parallel composition $P \| Q$ have standard readings.
The same applies to constructs for sending, receiving, selection, branching, and replicated receive prefixes, which are respectively denoted $\send{x}{y}.P$, $\recv{x}{y}.P$, $x \triangleleft \ell.P$, $x \triangleright \{i:P_i\}_{i \in I}$, and $\serv{x}{y}.P$.
Note that all these prefixes are \emph{blocking}, which means that our process calculus implements \emph{synchronous} communication (see, e.g.,~\cite{conf/csl/DeYoungCPT12,conf/ice/vdHeuvelP21,thesis/vdHeuvel24} for interpretations of linear logic as session type systems in the asynchronous setting).
Process $\fwd{x}{y}$ denotes a forwarder that ``fuses'' channels $x$ and $y$;
it is akin to the \emph{link processes} used in encodings of name passing using internal mobility~\cite{journal/tcs/Boreale98}.

We consider also constructs $\send{x}{}.P$ and $\recv{x}{}.P$, which specify the explicit closing of channels:
their synchronization represents the explicit deallocation of linear resources.
We use these constructs to give a \emph{non-silent} interpretation of $\1$ and $\bot$ (see \secref{ssec:types}).

In $\nu{y}{P}$, $\recv{x}{y}.P$, and $\serv{x}{y}.P$ the occurrence of $y$ is binding, with scope $P$\beginbas; we say that $y$ is \emph{bound} in $P$.
Non-bound occurrences of names are considered \emph{free}.
We consistently adopt Barendregt's naming convention: bound names are pairwise distinct and distinct from free names.
The sets of bound and free names of a process $P$ are denoted $\bn(P)$ and $\fn(P)$, respectively. \endbas
We identify processes up to consistent renaming of bound names, writing $\equiv_\alpha$ for this congruence.
We write $P\subst{y/x}$ for the capture-avoiding substitution of free occurrences of $y$ for $x$ in $P$.

\paragraph{Structural Congruence}

The following is an important notation of syntactical equivalence for \pull processes.

\begin{definition}\label{def:structuralcongruence}
    \emph{Structural congruence} is a binary relation on process terms, denoted $P \equiv Q$.
    It is defined as the least congruence on processes (i.e., closed under arbitrary process contexts) that satisfies the axioms in \Cref{f:red} (top).
\end{definition}

{The above formulation of structural congruence internalizes the conditions induced by typing (i.e., proof transformations in \ull), necessary for the correctness properties proven in \Cref{ssec:procprop}.}

\paragraph{Computation}

As is usual for Curry-Howard correspondences, computation is related to cut reduction.
Cuts are used in logic to combine two proofs that contain dual propositions.
As we will see, a cut in session type interpretations of linear logic entails the parallel composition and connection of two processes.
 {Generally, cut reduction transforms cuts into cuts on smaller types.}
In correspondences between linear logic and session types, cut reduction corresponds to \emph{communication}---the notion of computation in the $\pi$-calculus, which expresses internal behavior of processes.

The definition of the reduction relation follows.
Note that, for a simpler presentation, it includes a set of rules for commuting conversions ($\kappa$-rules);
an alternative presentation could treat commuting conversions as a behavioral equivalence.

\begin{figure}[p]
    \begin{mathpar}
        \begin{bussproof}[cutSymm]
            \bussAx{
                \nu{x} ( P \| Q ) \equiv \nu{x} ( Q \| P )
            }
        \end{bussproof}
        \and
        \begin{bussproof}[cutAssocL]
            \bussAssume{
                x \notin \fn(Q)
            }
            \bussAssume{
                y \notin \fn(P)
            }
            \bussBin{
                \nu{x} ( P \| \nu{y} ( Q \| R ) ) \equiv \nu{y} ( Q \| \nu{x} ( P \| R ) )
            }
        \end{bussproof}
        \and
        \begin{bussproof}[cutAssocR]
            \bussAssume{
                x \notin \fn(R)
            }
            \bussAssume{
                y \notin \fn(P)
            }
            \bussBin{
                \nu{x} ( P \| \nu{y} ( Q \| R ) ) \equiv \nu{y} ( \nu{x} ( P \| Q ) \| R )
            }
        \end{bussproof}
    \end{mathpar}
    \begin{mathpar}
        \and
        \begin{bussproof}[$\beta$id]
            \bussAssume{
                x \neq y
            }
            \bussUn{
                \nu{x}(P \| \fwd{x}{y}) \red P\subst{y/x}
            }
        \end{bussproof}
        \and
        \begin{bussproof}[$\beta$close]
            \bussAx{
                \nu{x} ( \send{x}{}.\0 \| \recv{x}{}.Q ) \red Q
            }
        \end{bussproof}
        \and
        \begin{bussproof}[$\beta$send]
            \bussAx{
                \nu{x} ( \nu{y} \send{x}{y} . ( P_1 \| P_2 ) \| \recv{x}{z}.Q ) \red \nu{x} ( P_2 \| \nu{y} ( P_1 \| Q\subst{y/z} ) )
            }
        \end{bussproof}
        \and
        \begin{bussproof}[$\beta$sel]
            \bussAssume{
                j \in I
            }
            \bussUn{
                \nu{x} ( x \triangleleft j.P \| x \triangleright \{i:Q_i\}_{i \in I} ) \red \nu{x} ( P \| Q_j )
            }
        \end{bussproof}
        \and
        \begin{bussproof}[$\beta$serv]
            \bussAx{
                \nu{x} ( \nu{y} \send{x}{y}.P \| \serv{x}{z}.Q ) \red \nu{x} ( \nu{y} ( P \| Q\subst{y/z} ) \| \serv{x}{z}.Q )
            }
        \end{bussproof}
        \and
        \beginbas
        \begin{bussproof}[$\beta$weaken]
            \bussAssume{
                x \notin \fn(P)
            }
            \bussUn{
                \nu{x} ( P \| \serv{x}{z}.Q ) \red P
            }
        \end{bussproof}
        \endbas
    \end{mathpar}
    \begin{mathpar}
        \and
        \begin{bussproof}[$\kappa$close]
            \bussAx{
                \nu{y} ( P \| \recv{x}{}.Q ) \red \recv{x}{}.\nu{y} ( P \| Q )
            }
        \end{bussproof}
        \and
        \begin{bussproof}[$\kappa$sendR]
            \bussAssume{
                y \in \fn(Q_2)
            }
            \bussUn{
                \nu{y} ( P \| \nu{z} \send{x}{z}.( Q_1 \| Q_2 ) ) \red \nu{z} \send{x}{z}.( Q_1 \| \nu{y} ( P \| Q_2 ) )
            }
        \end{bussproof}
        \and
        \begin{bussproof}[$\kappa$sendL]
            \bussAssume{
                y \in \fn(Q_1)
            }
            \bussUn{
                \nu{y} ( P \| \nu{z} \send{x}{z}.( Q_1 \| Q_2 ) ) \red \nu{z} \send{x}{z}.( \nu{y} ( P \| Q_1 ) \| Q_2 )
            }
        \end{bussproof}
        \and
        \begin{bussproof}[$\kappa$recv]
            \bussAx{
                \nu{y} ( P \| \recv{x}{z}.Q ) \red \recv{x}{z}.\nu{y} ( P \| Q )
            }
        \end{bussproof}
        \and
        \begin{bussproof}[$\kappa$sel]
            \bussAx{
                \nu{y} ( P \| x \triangleleft \ell.Q ) \red x \triangleleft \ell.\nu{y} ( P \| Q )
            }
        \end{bussproof}
        \and
        \begin{bussproof}[$\kappa$bra]
            \bussAx{
                \nu{y} ( x \triangleright \{i:P_i\}_{i \in I} \| Q ) \red x \triangleright \{i:\nu{y} ( P_i \| Q )\}_{i \in I}
            }
        \end{bussproof}
    \end{mathpar}
    \caption{Structural congruence and reduction for \pull.\label{f:red}}
\end{figure}

\begin{definition}\label{def:reductionrelation}
    \emph{Reduction} is a binary relation on process terms, denoted $P \red Q$, defined in \Cref{f:red}.
    It is closed under the following rules:
    \begin{mathpar}
        \begin{bussproof}
            \bussAssume{
                Q \red Q'
            }
            \bussUn[\ruleLabel{par}]{
                P \| Q \red P \| Q'
            }
        \end{bussproof}
        \and
        \begin{bussproof}
            \bussAssume{
                P \red Q
            }
            \bussUn[\ruleLabel{res}]{
                \nu{y}P \red \nu{y}Q
            }
        \end{bussproof}
        \and
        \begin{bussproof}
            \bussAssume{
                P \equiv P'
            }
            \bussAssume{
                P' \red Q'
            }
            \bussAssume{
                Q' \equiv Q
            }
            \bussTern[\ruleLabel{sc}]{
                P \red Q
            }
        \end{bussproof}
    \end{mathpar}
    We write $\red_\beta$ to denote reductions that follow from $\beta$-rules.
\end{definition}

\smallskip
\noindent
In this definition, Rule~\ruleLabel{$\beta$id} replaces a channel $x$ with a channel $y$ which $x$ is forwarded to.
Rule~\ruleLabel{$\beta$close} formalizes the explicit channel deallocation mentioned above: the synchronization between an empty send and an empty receive effectively closes channel $x$.
Rule~\ruleLabel{$\beta$send} formalizes the synchronization between a send and a corresponding receive, substituting the sent name for the received name in the continuation of the receive.
Rule~\ruleLabel{$\beta$sel} formalizes the synchronization between a selection and a branch.
Rule~\ruleLabel{$\beta$serv} allows to synchronize a client and a service; a copy of the service $Q$ is spawned with the sent name substituted for the received name while the replicated receive remains.
Rule~\ruleLabel{$\beta$weaken} allows cleaning up a service if it has no clients.

The $\kappa$-rules reflect \emph{commuting conversions} in linear logic, which in the process calculus allow pulling prefixes on free names out of series of consecutive restrictions;
as we will see, we use them such that we can state a general progress result.

Rules~\ruleLabel{par}, \ruleLabel{res}, and~\ruleLabel{sc} close reduction under parallel composition, restriction and structural congruence, respectively (\Cref{def:structuralcongruence}).

\paragraph{Type Checking}

The inference system of \ull is a sequent calculus with sequents of the form $\Gamma; \Delta \vdash \Lambda$.
Here, $\Gamma$, $\Delta$ and $\Lambda$ denote \emph{regions} which collect propositions and obey different structural criteria.
$\Gamma$ is the \emph{unrestricted} region, which contains propositions that can be indefinitely used.
$\Delta$ and $\Lambda$ are the \emph{linear} regions, which contain propositions that must be used exactly once.
We write `$\emptyset$' to denote an empty region.
Also, we extend duality to regions:
$\dual{(\Delta)}$ contains exactly the duals of the propositions in $\Delta$.

The type system for \pull is an extension of \ull's inference system with process and channel name annotations on sequents, such that judgments are of the form $\Gamma; \Delta \vdash P :: \Lambda$.
Then, the regions $\Gamma$, $\Delta$, and $\Lambda$ denote the unrestricted resp.\ linear \emph{contexts} of $P$, consisting of assignments~$x:A$ where $x$ is a channel name and $A$ is a proposition/type.

\begin{figure}[p]
    \begin{mathpar}
        \begin{bussproof}[idR][$\ast$]
            \bussAx{
                \Gamma; x:A \vdash \fwd{x}{y} :: y:A
            }
        \end{bussproof}
        \and
        \begin{bussproof}[idL]
            \bussAx{
                \Gamma; x:A, y:\dual{A} \vdash \fwd{x}{y} :: \emptyset
            }
        \end{bussproof}
        \and
        \begin{bussproof}[$\1$R][$\ast$]
            \bussAx{
                \Gamma; \emptyset \vdash \send{x}{}.\0 :: x:\1
            }
        \end{bussproof}
        \and
        \begin{bussproof}[$\1$L][$\ast$]
            \bussAssume{
                \Gamma; \Delta \vdash P :: \Lambda
            }
            \bussUn{
                \Gamma; \Delta, x:\1 \vdash \recv{x}{}.P :: \Lambda
            }
        \end{bussproof}
        \and
        \begin{bussproof}[$\bot$R]
            \bussAssume{
                \Gamma; \Delta \vdash P :: \Lambda
            }
            \bussUn{
                \Gamma; \Delta \vdash \recv{x}{}.P :: \Lambda, x:\bot
            }
        \end{bussproof}
        \and
        \begin{bussproof}[$\bot$L]
            \bussAx{
                \Gamma; x:\bot \vdash \send{x}{}.\0 :: \emptyset
            }
        \end{bussproof}
        \and
        \begin{bussproof}[$\tensor$R][$\ast$]
            \bussAssume{
                \Gamma; \Delta \vdash P :: \Lambda, y:A
            }
            \bussAssume{
                \Gamma; \Delta' \vdash Q :: \Lambda', x:B
            }
            \bussBin{
                \Gamma; \Delta, \Delta' \vdash \nu{y} \send{x}{y}.(P \| Q) :: \Lambda, \Lambda', x:A \tensor B
            }
        \end{bussproof}
        \and
        \begin{bussproof}[$\tensor$L][$\ast$]
            \bussAssume{
                \Gamma; \Delta, y:A, x:B \vdash P :: \Lambda
            }
            \bussUn{
                \Gamma; \Delta, x:A \tensor B \vdash \recv{x}{y}.P :: \Lambda
            }
        \end{bussproof}
        \and
        \begin{bussproof}[$\parr$R]
            \bussAssume{
                \Gamma; \Delta \vdash P :: \Lambda, y:A, x:B
            }
            \bussUn{
                \Gamma; \Delta \vdash \recv{x}{y}.P :: \Lambda, x:A \parr B
            }
        \end{bussproof}
        \and
        \begin{bussproof}[$\parr$L]
            \bussAssume{
                \Gamma; \Delta, y:A \vdash P :: \Lambda
            }
            \bussAssume{
                \Gamma; \Delta', x:B \vdash Q :: \Lambda'
            }
            \bussBin{
                \Gamma; \Delta, \Delta', x:A \parr B \vdash \nu{y} \send{x}{y}.(P \| Q) :: \Lambda, \Lambda'
            }
        \end{bussproof}
        \and
        \begin{bussproof}[$\lolli$R][$\ast$]
            \bussAssume{
                \Gamma; \Delta, y:A \vdash P :: \Lambda, x:B
            }
            \bussUn{
                \Gamma; \Delta \vdash \recv{x}{y}.P :: \Lambda, x:A \lolli B
            }
        \end{bussproof}
        \and
        \begin{bussproof}[$\lolli$L][$\ast$]
            \bussAssume{
                \Gamma; \Delta \vdash P :: \Lambda, y:A
            }
            \bussAssume{
                \Gamma; \Delta', x:B \vdash Q :: \Lambda'
            }
            \bussBin{
                \Gamma; \Delta, \Delta', x:A \lolli B \vdash \nu{y}
                \send{x}{y}.(P \| Q) :: \Lambda, \Lambda'
            }
        \end{bussproof}
        \and
        \begin{bussproof}[$\oplus$R][$\ast$]
            \bussAssume{
                \Gamma; \Delta \vdash P :: \Lambda, x:A_j
            }
            \bussAssume{
                j \in I
            }
            \bussBin{
                \Gamma; \Delta \vdash x \triangleleft j .P :: \Lambda, x:\oplus\{i:A_i\}_{i \in I}
            }
        \end{bussproof}
        \and
        \begin{bussproof}[$\oplus$L][$\ast$]
            \bussAssume{
                \forall i \in I.~ \Gamma; \Delta, x:A_i \vdash P_i :: \Lambda
            }
            \bussUn{
                \Gamma; \Delta, x:\oplus\{i:A_i\}_{i \in I} \vdash x \triangleright \{i:P_i\}_{i \in I} :: \Lambda
            }
        \end{bussproof}
        \and
        \begin{bussproof}[$\&$R][$\ast$]
            \bussAssume{
                \forall i \in I.~ \Gamma; \Delta \vdash P_i :: \Lambda, x:A_i
            }
            \bussUn{
                \Gamma; \Delta \vdash x \triangleright \{i:P_i\}_{i \in I} :: \Lambda, x:\&\{i:A_i\}_{i \in I}
            }
        \end{bussproof}
        \and
        \begin{bussproof}[$\&$L][$\ast$]
            \bussAssume{
                \Gamma; \Delta, x:A_j \vdash P :: \Lambda
            }
            \bussAssume{
                j \in I
            }
            \bussBin{
                \Gamma; \Delta, x:\&\{i:A_i\}_{i \in I} \vdash x \triangleleft j .P :: \Lambda
            }
        \end{bussproof}
        \and
        \begin{bussproof}[copyR]
            \bussAssume{
                \Gamma, u:A; \Delta \vdash P :: \Lambda, x:\dual{A}
            }
            \bussUn{
                \Gamma, u:A; \Delta \vdash \nu{x} \send{u}{x}.P :: \Lambda
            }
        \end{bussproof}
        \and
        \begin{bussproof}[copyL][$\ast$]
            \bussAssume{
                \Gamma, u:A; \Delta, x:A \vdash P :: \Lambda
            }
            \bussUn{
                \Gamma, u:A; \Delta \vdash \nu{x} \send{u}{x}.P :: \Lambda
            }
        \end{bussproof}
        \and
        \begin{bussproof}[$\bang$R][$\ast$]
            \bussAssume{
                \Gamma; \emptyset \vdash P :: y:A
            }
            \bussUn{
                \Gamma; \emptyset \vdash \serv{x}{y}.P :: x:\bang A
            }
        \end{bussproof}
        \and
        \begin{bussproof}[$\bang$L][$\ast$]
            \bussAssume{
                \Gamma, u:A; \Delta \vdash P :: \Lambda
            }
            \bussUn{
                \Gamma; \Delta, x:\bang A \vdash P\subst{x/u} :: \Lambda
            }
        \end{bussproof}
        \and
        \begin{bussproof}[$\whynot$R]
            \bussAssume{
                \Gamma, u:A; \Delta \vdash P :: \Lambda
            }
            \bussUn{
                \Gamma; \Delta \vdash P\subst{x/u} :: \Lambda, x:\whynot \dual{A}
            }
        \end{bussproof}
        \and
        \begin{bussproof}[$\whynot$L]
            \bussAssume{
                \Gamma; y:A \vdash P :: \emptyset
            }
            \bussUn{
                \Gamma; x:\whynot A \vdash \serv{x}{y}.P :: \emptyset
            }
        \end{bussproof}
    \end{mathpar}

    \caption{The $\pull$ type system.}\label{f:ull-inf}
\end{figure}

\begin{figure}[t]
    \begin{mathpar}
        \begin{bussproof}[cutRL][$\ast$]
            \bussAssume{
                \Gamma; \Delta \vdash P :: \Lambda, x:A
            }
            \bussAssume{
                \Gamma; \Delta', x:A \vdash Q :: \Lambda'
            }
            \bussBin{
                \Gamma; \Delta, \Delta' \vdash \nu{x}(P \| Q) :: \Lambda, \Lambda'
            }
        \end{bussproof}
        \and
        \begin{bussproof}[cutLR][$\ast$]
            \bussAssume{
                \Gamma; \Delta, x:A \vdash P :: \Lambda
            }
            \bussAssume{
                \Gamma; \Delta' \vdash Q :: \Lambda', x:A
            }
            \bussBin{
                \Gamma; \Delta, \Delta' \vdash \nu{x}(P \| Q) :: \Lambda, \Lambda'
            }
        \end{bussproof}
        \and
        \begin{bussproof}[cutRR]
            \bussAssume{
                \Gamma; \Delta \vdash P :: \Lambda, x:A
            }
            \bussAssume{
                \Gamma; \Delta' \vdash Q :: \Lambda', x:\dual{A}
            }
            \bussBin{
                \Gamma; \Delta, \Delta' \vdash \nu{x}(P \| Q) :: \Lambda, \Lambda'
            }
        \end{bussproof}
        \and
        \begin{bussproof}[cutLL]
            \bussAssume{
                \Gamma; \Delta, x:A \vdash P :: \Lambda
            }
            \bussAssume{
                \Gamma; \Delta', x:\dual{A} \vdash Q :: \Lambda'
            }
            \bussBin{
                \Gamma; \Delta, \Delta' \vdash \nu{x}(P \| Q) :: \Lambda, \Lambda'
            }
        \end{bussproof}
        \and
        \begin{bussproof}[cut$\bang$R][$\ast$]
            \bussAssume{
                \Gamma, u:A; \Delta \vdash P :: \Lambda
            }
            \bussAssume{
                \Gamma; \emptyset \vdash Q :: x:A
            }
            \bussBin{
                \Gamma; \Delta \vdash \nu{u} ( P \| \serv{u}{x}.Q ) :: \Lambda
            }
        \end{bussproof}
        \and
        \begin{bussproof}[cut$\bang$L][$\ast$]
            \bussAssume{
                \Gamma; \emptyset \vdash P :: x:A
            }
            \bussAssume{
                \Gamma, u:A; \Delta \vdash Q :: \Lambda
            }
            \bussBin{
                \Gamma; \Delta \vdash \nu{u} ( \serv{u}{x}.P \| Q ) :: \Lambda
            }
        \end{bussproof}
        \and
        \begin{bussproof}[cut$\whynot$R]
            \bussAssume{
                \Gamma, u:A; \Delta \vdash P :: \Lambda
            }
            \bussAssume{
                \Gamma; x:\dual{A} \vdash Q :: \emptyset
            }
            \bussBin{
                \Gamma; \Delta \vdash \nu{u} ( P \| \serv{u}{x}.Q ) :: \Lambda
            }
        \end{bussproof}
        \and
        \begin{bussproof}[cut$\whynot$L]
            \bussAssume{
                \Gamma; x:\dual{A} \vdash P :: \emptyset
            }
            \bussAssume{
                \Gamma, u:A; \Delta \vdash Q :: \Lambda
            }
            \bussBin{
                \Gamma; \Delta \vdash \nu{u} ( \serv{u}{x}.P \| Q ) :: \Lambda
            }
        \end{bussproof}
    \end{mathpar}
    \caption{Cut-rules of the \pull type system.}\label{f:ull-inf-cut}
\end{figure}

\Cref{f:ull-inf,f:ull-inf-cut} give the typing rules of \pull, which are based directly on the linear fragment of LU in~\cite{journal/apal/Girard93} (some rules are marked with $\ast$, which we will refer to later).

We comment on the rules in \Cref{f:ull-inf}.
Axioms~\ruleLabel{idR} and~\ruleLabel{idL} type forwarding constructs which connect two channels of dual type.
Axioms~\ruleLabel{$\1$R} and~\ruleLabel{$\bot$L} type processes that close a session with an empty send after which they become inactive.
Rules~\ruleLabel{$\bot$R} and~\ruleLabel{$\1$L} type processes that close a session with an empty receive.
These four rules define a \emph{non-silent} interpretation for $\1$ and $\bot$ that entails process communication (cf.\ \Cref{def:reductionrelation}), which corresponds to cut reductions in proofs.
(An alternative \emph{silent} interpretation of $\1$ \beginbas and $\bot$ \endbas is discussed in \Cref{ss:silent}.)

The typing system elegantly induces processes under the \emph{internal} mobility discipline, whereby only fresh channels are exchanged in communications~\cite{journal/tcs/Sangiorgi96,journal/tcs/Boreale98}.
Rules~\ruleLabel{$\tensor$R}, \ruleLabel{$\parr$L}, and~\ruleLabel{$\lolli$L} type bound sends, where one process provides the sent channel and another independent process provides the continuation channel.
Rules~\ruleLabel{$\tensor$L}, \ruleLabel{$\parr$R}, and~\ruleLabel{$\lolli$R} type receive-prefixed processes.
Rules~\ruleLabel{$\oplus$R} and~\ruleLabel{$\&$L} type selection and rules~\ruleLabel{$\oplus$L} and~\ruleLabel{$\&$R} type branching.

 {Our interpretation of $\bang A$ and $\whynot A$ as server and client behaviors follows the   interpretation of classical linear logic in~\cite{journal/mscs/CairesPT16}.}
Rules~\ruleLabel{copyR} and~\ruleLabel{copyL} type clients that connect to a service by sending a fresh channel.
Rules~\ruleLabel{$\bang$R} and~\ruleLabel{$\whynot$L} allow the typing of unused services and Rules~\ruleLabel{$\bang$L} and~\ruleLabel{$\whynot$R} allow adding unused services to the unrestricted context.

\plscheck{
    \Cref{f:ull-inf-cut} gives a series of so-called \emph{cut-rules}, that type channel connections.
    The number and shape of these rules is a difference with respect to previous presentations.
    Rules~\ruleLabel{cutRL}, \ruleLabel{cutLR}, \ruleLabel{cutLL}, and~\ruleLabel{cutRR} type pairs of processes that have a channel of dual type in common by composing them in parallel and immediately binding their common channel.
    The four similar rules provide for all possible sides the cut channel can appear on.
}
This way, constructs for restriction and parallel composition are jointly treated.
Rules~\ruleLabel{cut$\bang$R}, \ruleLabel{cut$\bang$L}, \ruleLabel{cut$\whynot$R}, and~\ruleLabel{cut$\whynot$R} type the connection of a service provider with potential clients; a process $Q$ with potential clients has a channel $u$ in its unrestricted context, so the rules create a service from a process $P$ that has a single channel $x$ of type dual to $u$'s type by prefixing it with replicated reception on $u$ (forming $\serv{u}{x}.P$) and then composing this process in parallel with $Q$ and binding $u$.

\plscheck{This abundance of cut-rules is derived from the generality of \ull's judgments and necessary for proving the correctness results presented in \Cref{ssec:procprop}.
In \Cref{ssec:duality} we shall consider an alternative presentation of \pull, which allows moving channels between the left- and right-hand regions of typing judgments using duality; as we will see, in such a presentation we will be able to drastically cut down the number of cut-rules.}

\paragraph{Differences with LU}

{As already mentioned, for the purposes of our formal comparison we consider a linear logic derived from LU~\cite{journal/apal/Girard93} restricted to linear connectives.
The following are notable differences between our linear logic and the linear fragment of LU:}
\begin{itemize}[nosep]
    \item
        we include a Rule~\ruleLabel{idL} which is complementary to~\ruleLabel{idR};
    \item
        we include Rules~\ruleLabel{$\1$L} and~\ruleLabel{$\bot$R} which are lacking in LU;
    \item
        we omit rules for $\top$ and $\0$ (the units of $\&$ and $\oplus$, resp.), which are usually disregarded in session type interpretations of linear logic (an exception is \cite{conf/places/HorneP23}, which uses $\top$ and $\0$ to give a logical account of subtyping);
    \item
        we omit rules that move propositions between the left and right linear regions using duality (in \Cref{ssec:duality} we will return to these rules);
    \item
       {because the order of assumptions in typing rules makes a practical difference, we include additional symmetric cut-rules.}
\end{itemize}

\subsection{Correctness Properties}
\label{ssec:procprop}

Session type systems for the $\pi$-calculus derived from the Curry-Howard correspondence enforce strong correctness properties for processes, which  follow directly from properties of the logic, in particular from cut elimination.
This is no different for \pull.
    Our first result is the \emph{safety} property of \emph{subject congruence and reduction} (\Cref{thm:subjectcongruence,thm:subjectreduction}), which says that typability is consistent across structural congruence and reductions.

\begin{theorem}[Subject Congruence]\label{thm:subjectcongruence}
    If $\Gamma ; \Delta \vdash P :: \Lambda$ and $P \equiv Q$, then $\Gamma ; \Delta \vdash Q :: \Lambda$.
\end{theorem}

\begin{proof}[Proof (Sketch)]
    By induction on the derivation of the structural congruence.
    The only inductive case is the closure under arbitrary process contexts, which follows from the IH directly.
    The base cases correspond to a rule in \Cref{f:red} (top).
    In each case, we infer the typing of $P$ and $Q$ from the shapes of the processes in the rule, and show that these typing inferences have identical assumptions and conclusion.
    The cases of Rules~\ruleLabel{cutAssocL} and~\ruleLabel{cutAssocR} are straightforward as usual.
    The analysis of the more interesting case of Rule~\ruleLabel{cutSymm} depends on the last-applied cut-rule.
    If the left-hand side uses, e.g., Rule~\ruleLabel{cutLR}, then the right-hand side should use Rule~\ruleLabel{cutRL}.
\end{proof}

\begin{theorem}[Subject Reduction]\label{thm:subjectreduction}
    If $\Gamma; \Delta \vdash P :: \Lambda$ and $P \red Q$, then $\Gamma; \Delta \vdash Q :: \Lambda$.
\end{theorem}

\begin{proof}[Proof (Sketch)]
    By induction on the derivation of the reduction.
    The cases correspond to the rules in \Cref{f:red} (bottom), as well as the closure rules in \Cref{def:reductionrelation}.
    In each case, we infer the typing of $P$ and construct \beginbas one \endbas for $Q$ from the shapes of the processes in the rule, and show that these typing inferences have identical assumptions and conclusions.
    We detail two cases.

    The case of Rule~\ruleLabel{$\beta$serv} serves to illustrate the need for multiple symmetric cut-rules.
    Suppose, for example, that the last-applied rule is Rule~\ruleLabel{cut$\bang$R}, and that the client request is derived using Rule~\ruleLabel{copyL}:
    \[
        \begin{bussproof}
            \bussAssume{
                \Gamma , x:A ; \Delta , y:A \vdash P :: \Lambda
            }
            \bussUn[\ruleLabel{copyL}]{
                \Gamma , x:A ; \Delta \vdash \nu{y} \send{x}{y}.P :: \Lambda
            }
            \bussAssume{
                \Gamma ; \emptyset \vdash Q :: z:A
            }
            \bussBin[\ruleLabel{cut$\bang$R}]{
                \Gamma ; \Delta \vdash \nu{x} ( \nu{y} \send{x}{y}.P \| \serv{x}{z}.Q ) :: \Lambda
            }
        \end{bussproof}
    \]
    As per Rule~\ruleLabel{$\beta$serv}, we need to identically type $\nu{x} ( \nu{y} ( P \| Q\subst{y/z} ) \| \serv{x}{z}.Q )$ using the same assumptions as above.
    Had we only had, e.g., Rules~\ruleLabel{cutRL} and~\ruleLabel{cutLL}, this would not be possible.
    However, with the rules in \Cref{f:ull-inf-cut} it is no problem.
    We first have to add $x:A$ into the persistent regions of the proof of the typing of $Q$, and substitute $y$ for $z$, after which we derive the following:
    \[
        \begin{bussproof}
            \bussAssume{
                \Gamma , x:A ; \Delta , y:A \vdash P :: \Lambda
            }
            \bussAssume{
                \Gamma , x:A ; \emptyset \vdash Q\subst{y/z} :: y:A
            }
            \bussBin[\ruleLabel{cutLR}]{
                \Gamma , x:A ; \Delta \vdash \nu{y} ( P \| Q\subst{y/z} ) :: \Lambda
            }
            \bussAssume{
                \Gamma ; \emptyset \vdash Q :: z:A
            }
            \bussBin[\ruleLabel{cut$\bang$R}]{
                \Gamma ; \Delta \vdash \nu{x} ( \nu{y} ( P \| Q\subst{y/z} ) \| \serv{x}{z}.Q ) :: \Lambda
            }
        \end{bussproof}
    \]

    As another representative case, we consider Rule~\ruleLabel{$\beta$send}.
    We have $P = \nu{x} ( \nu{y}\send{x}{y}.(R \| S) \| x(z).T) \red \nu{x} ( S \| \nu{y} ( R \| T\subst{y/z} ) ) = Q$.
    There are multiple ways to type $P$, depending on the cut-rule applied.
    Here, we give the example of Rule~\ruleLabel{cutR}.
    The proof of $\Gamma; \Delta \vdash P :: \Lambda$ looks as follows:
    \[
        \begin{bussproof}
            \bussAssume{
                \Gamma; \Delta_1 \vdash R :: \Lambda_1, y:A
            }
            \bussAssume{
                \Gamma; \Delta_2 \vdash S :: \Lambda_2, x:B
            }
            \bussBin[\ruleLabel{$\tensor$R}]{
                \Gamma; \Delta_1, \Delta_2 \vdash \nu{y}\send{x}{y}.(R \| S) :: \Lambda_1, \Lambda_2, x:A \tensor B
            }
            \bussAssume{
                \Gamma; \Delta_3, z:A, x:B \vdash T :: \Lambda_3
            }
            \bussUn[\ruleLabel{$\tensor$L}]{
                \Gamma; \Delta_3, x:A \tensor B \vdash x(z).T :: \Lambda_3
            }
            \bussBin[\ruleLabel{cutR}]{
                \Gamma; \underbrace{\Delta_1, \Delta_2, \Delta_3}_{\Delta} \vdash \underbrace{\nu{x}(\nu{y}\send{x}{y}(R \| S) \| T)}_{P} :: \underbrace{\Lambda_1, \Lambda_2, \Lambda_3}_{\Lambda}
            }
        \end{bussproof}
    \]
    We can then construct a proof of $\Gamma; \Delta \vdash Q :: \Lambda$ using the assumptions in the above proof.
    \[
        \begin{bussproof}
            \bussAssume{
                \Gamma; \Delta_2 \vdash S :: \Lambda_2, x:B
            }
            \bussAssume{
                \Gamma; \Delta_1 \vdash R :: \Lambda_1, y:A
            }
            \bussAssume{
                \Gamma; \Delta_3, y:A, x:B \vdash T \subst{y/z} :: \Lambda_3
            }
            \bussBin[\ruleLabel{cutR}]{
                \Gamma; \Delta_1, \Delta_3, x:B \vdash \nu{y}(R \| T \subst{y/z} ) :: \Lambda_1, \Lambda_3
            }
            \bussBin[\ruleLabel{cutR}]{
                \Gamma; \Delta \vdash \underbrace{\nu{x}( S \| \nu{y}(R \| T \subst{y/z}) )}_{Q} :: \Lambda
            }
        \end{bussproof}
    \]
\end{proof}

Our next result is \emph{progress},
the \beginbas \emph{safety} \endbas property that says that the specific form of composition and restriction in \pull following from the cut-rule enables communication, and that processes never get stuck waiting for each other:

\begin{theorem}[Progress]\label{thm:progress}
    If $\Gamma; \Delta \vdash P :: \Lambda$ and $P \equiv \nu{x}(Q \| R)$, then there exists $P'$ such that $P \red P'$.
\end{theorem}

\begin{proof}
    By induction on the size of the proof of $\Gamma; \Delta \vdash P :: \Lambda$.
    By \Cref{thm:subjectcongruence}, $\Gamma ; \Delta \vdash \nu{x} ( Q \| R ) :: \Lambda$.
    By assumption, the last inference of the derivation thereof is either a linear cut or an unrestricted cut.

    \textbf{(Case linear cut)}
    The last-applied rule can be~\ruleLabel{cutR} or~\ruleLabel{cutL}.
    W.l.o.g.\ assume the former.
    By inversion of \ruleLabel{cutR}, we have a proof $\pi_Q$ of $\Gamma; \Delta_Q \vdash Q :: \Lambda_Q, x:A$ and a proof $\pi_R$ of $\Gamma; \Delta_R, x:A \vdash R :: \Lambda_R$ where $\Delta_Q, \Delta_R = \Delta$ and $\Lambda_Q, \Lambda_R = \Lambda$.

    If the last-applied rules in $\pi_Q$ and $\pi_R$ are both on $x$, then we apply a $\beta$-reduction depending on~$A$.
    For example, assume $A = B \tensor C$.
    Then the last-applied rules in $\pi_Q$ and $\pi_R$ are \ruleLabel{$\tensor$R} and \ruleLabel{$\tensor$L}, respectively.
    Hence, by Rule~\ruleLabel{$\beta$send}, $P \equiv \nu{x}(\nu{y}\send{x}{y}.(Q_y \| Q_x) \| \recv{x}{y}.R') \red \nu{x}(Q_x \| \nu{y}(Q_y \| R'))$.

    Otherwise, w.l.o.g.\ assume the last-applied rule not on $x$ is in $\pi_Q$.
    Then, if $Q$ is a cut, by the induction hypothesis, $Q \red Q'$, and hence $P \equiv \nu{x}(Q \| R) \red \nu{x}(Q' \| R)$.
    Otherwise, $Q$ is prefixed by an action on some free channel $y$ which is not a free channel of $R$.
    Hence, we apply a $\kappa$-conversion depending on the type of the channel the last-applied rule in $\pi_Q$ works on.
    For example, if this rule introduces $y:B \tensor C$ on the right, then, by Rule~\ruleLabel{$\kappa{\tensor}$}, $P \equiv \nu{x}(\nu{z}\send{y}{z}.(Q_y \| Q_{z,x}) \| R_x) \red \nu{z}\send{y}{z}.(\nu{x}(Q_{z,x} \| R_x) \| Q_y)$.

    \textbf{(Case unrestricted cut)}
    The last-applied rule can be \ruleLabel{cut$\bang$} or \ruleLabel{cut$\whynot$}.
    W.l.o.g.\ assume the former.
    Then $Q \equiv \serv{x}{y}.Q'$, and, by inversion of this rule, we have a proof $\pi_{Q'}$ of $\Gamma; \emptyset \vdash Q' :: y:A$ and a proof $\pi_R$ of $\Gamma, x:A; \Delta \vdash R :: \Lambda$.
    If $x \notin \fn(R)$, then, by Rule~\ruleLabel{$\beta$weaken}, $P \equiv \nu{x}( R \| \serv{x}{y}.Q' ) \red R$.
    Otherwise, the next step depends on the last-applied rule in $\pi_R$.

    If the last-applied rule in $\pi_R$ is on $x$, then it must be \ruleLabel{copyR} or \ruleLabel{copyL}.
    W.l.o.g.\ assume the former.
    Then $R \equiv \nu{y}\send{x}{y}.R'$, so, by Rule~\ruleLabel{$\beta$serv}, $P \equiv \nu{x}(\nu{y}\send{x}{y}.R' \| \serv{x}{y}.Q') \red \nu{x}(\nu{y}(Q' \| R') \| \serv{x}{y}.Q')$.

    Otherwise, the proof proceeds as in the last part of the case of linear cut.
\end{proof}

\beginbas
It follows from progress that \emph{closed} processes are \emph{deadlock free}.
In a closed process, all channels are connected, except for one channel.
This remaining channel, which the process simply closes, can be seen as a \emph{barb} that signals successful termination.
As such, deadlock freedom ensures that closed processes that cannot reduce are successfully terminated and thus not stuck in a deadlock.
Let us write $P \nred$ to denote that there is no $P'$ such that $P \red P'$.

\begin{theorem}[Deadlock Freedom]
    \label{thm:df}
    Suppose $\emptyset ; \emptyset \vdash P :: z : \1$ or $\emptyset ; z : \bot \vdash P :: \emptyset$.
    If $P \nred$, then $P \equiv \send{z}{} . \0$.
\end{theorem}

\begin{proof}
    Towards contradiction, assume $P \not\equiv \send{z}{} . \0$.
    The only well-typed possibility is that $P \equiv \nu{x} ( Q \| R )$.
    By \Cref{thm:progress}, then $P \red P'$.
    This contradicts the assumption that $P \nred$, proving the thesis.
\end{proof}

The combination of \Cref{thm:subjectreduction,thm:progress} allows us to adapt existing techniques (e.g., the logical relations of~\cite{conf/esop/PerezCPT12}) to prove \emph{termination}: the \emph{liveness} property that says that well-typed processes eventually stop reducing.
\endbas

\subsection{On Duality}
\label{ssec:duality}

It may seem that there is an extensive redundancy in the system of rules in \Cref{f:ull-inf}, caused by the two sidedness of \pull's judgments:
every connective can be inferred on either side of judgments.
For example, Rules~\ruleLabel{$\tensor$R}, \ruleLabel{$\parr$L}, and~\ruleLabel{$\lolli$L} all type the send of a channel;
which rule to use depends on the side of the judgment the involved channels are on.
However, there is no actual redundancy, for if we were to omit rules for, e.g., $\parr$ and $\lolli$, it would be impossible to type a send on a previously received channel.

This abundance of typing rules in \pull can be explained by its full support for duality:
for every rule inferring a connective on one side of a judgment, there is \beginbas a \endbas rule for inferring the connective's dual on the other side of a judgment.
To make this duality explicit, we define an alternative type system by restricting \pull's rules to a specific fragment and adding LU's rules for \emph{moving propositions} between sides of judgments:

\begin{definition}\label{def:pullm}
    The type system \pullm, with judgments $\Gamma; \Delta \vpull P :: \Lambda$, is defined on the process calculus as defined in \Cref{ssec:proccal}.
    Its rules are the $\ast$-marked rules in \Cref{f:ull-inf} plus the following rules:
    \begin{mathpar}
        \begin{bussproof}[$\mleft$]
            \bussAssume{
                \Gamma; \Delta \vpull P :: \Lambda, x:A
            }
            \bussUn{
                \Gamma; \Delta, x:\dual{A} \vpull P :: \Lambda
            }
        \end{bussproof}
        \and
        \begin{bussproof}[$\mright$]
            \bussAssume{
                \Gamma; \Delta, x:A \vpull P :: \Lambda
            }
            \bussUn{
                \Gamma; \Delta \vpull P :: \Lambda, x:\dual{A}
            }
        \end{bussproof}
    \end{mathpar}
\end{definition}

\smallskip
\noindent
Fortunately, in the presence of these two rules, a number of other rules become truly redundant:
all rules in \Cref{f:ull-inf} not marked with $\ast$ are admissible or derivable in \pullm.
Dually, Rules~\ruleLabel{$\mleft$} and~\ruleLabel{$\mright$} are admissible in vanilla \pull.
The following theorem formalizes these facts:

\begin{theorem}\label{thm:madmissible}
    ~
    \begin{enumerate}
        \item
            The rules in \Cref{f:ull-inf,f:ull-inf-cut} not marked with $\ast$ are admissible or derivable in \pullm, and
        \item
            Rules~\ruleLabel{$\mleft$} and~\ruleLabel{$\mright$}, as given in \Cref{def:pullm}, are admissible in \pull.
    \end{enumerate}
\end{theorem}

\begin{proof}
    \emph{(Item 1)}
    Suppose given a proof of $\Gamma; \Delta \vpull P :: \Lambda$.
    By applying induction on the structure of this proof we show that any applications of non-$\ast$-marked rules can be replaced with applications of $\ast$-marked rules in combination with uses of Rules~\ruleLabel{$\mleft$} or~\ruleLabel{$\mright$}.
    We discuss every possible last-applied rule, omitting cases of $\ast$-marked rules as they follow directly from the induction hypothesis.

    \begin{align*}
        & \bullet \ruleLabel{idL}
        && 1\quad
        && \Gamma; x:A, y:\dual{A} \vpull \fwd{x}{y} :: \emptyset
        \tag{assumption} \\
        & && 2\quad
        && \Gamma; x:A \vpull \fwd{x}{y} :: y:A
        \tag{\ruleLabel{idR}} \\
        & && 3\quad
        && \Gamma; x:A, y:\dual{A} \vpull \fwd{x}{y} :: \emptyset
        \tag{\ruleLabel{$\mleft$} on 2}
        \displaybreak[1] \\[5pt]
        & \bullet \ruleLabel{$\bot$R}
        && 1\quad
        && \Gamma; \Delta \vpull \recv{x}{}.P :: \Lambda, x:\bot
        \tag{assumption} \\
        & && 2\quad
        && \Gamma; \Delta \vpull P :: \Lambda
        \tag{inversion on 1} \\
        & && 3\quad
        && \Gamma; \Delta \vpull P :: \Lambda ~\text{with only $\ast$ rules}
        \tag{IH on 2} \\
        & && 4\quad
        && \Gamma; \Delta, x:\1 \vpull P :: \Lambda
        \tag{\ruleLabel{$\1$L} on 3} \\
        & && 5\quad
        && \Gamma; \Delta \vpull P :: \Lambda, x:\bot
        \tag{\ruleLabel{$\mright$} on 4}
        \displaybreak[1] \\[5pt]
        & \bullet \ruleLabel{$\bot$L}
        && 1\quad
        && \Gamma; x:\bot \vpull \send{x}{}.\0 :: \emptyset
        \tag{assumption} \\
        & && 2\quad
        && \Gamma; \emptyset \vpull \send{x}{}.\0 :: x:\1
        \tag{\ruleLabel{$\1$R}} \\
        & && 3\quad
        && \Gamma; x:\bot \vpull \send{x}{}.\0 :: \emptyset
        \tag{\ruleLabel{$\mleft$} on 2}
        \displaybreak[1] \\[5pt]
        & \bullet \ruleLabel{$\parr$R}
        && 1\quad
        && \Gamma; \Delta \vpull \recv{x}{y}.P :: \Lambda, x:A \parr B
        \tag{assumption} \\
        & && 2\quad
        && \Gamma; \Delta \vpull P :: \Lambda, y:A, x:B
        \tag{inversion on 1} \\
        & && 3\quad
        && \Gamma; \Delta \vpull P :: \Lambda, y:A, x:B ~\text{with only $\ast$ rules}
        \tag{IH on 2} \\
        & && 4\quad
        && \Gamma; \Delta, y:\dual{A}, x:\dual{B} \vpull P :: \Lambda
        \tag{\ruleLabel{$\mleft$} twice on 3} \\
        & && 5\quad
        && \Gamma; \Delta, x:\dual{A} \tensor \dual{B} \vpull \recv{x}{y}.P :: \Lambda
        \tag{\ruleLabel{$\tensor$L} on 4} \\
        & && 6\quad
        && \Gamma; \Delta \vpull \recv{x}{y}.P :: \Lambda, x:A \parr B
        \tag{\ruleLabel{$\mright$} on 5}
        \displaybreak[1] \\[5pt]
        & \bullet \ruleLabel{$\parr$L}
        && 1\quad
        && \Gamma; \Delta, \Delta', x:A \parr B \vpull \nu{y}\send{x}{y}.(P \| Q) :: \Lambda, \Lambda'
        \tag{assumption} \\
        & && 2\quad
        && \Gamma; \Delta, y:A \vpull P :: \Lambda \\
        & && 3\quad
        && \Gamma; \Delta', x:B \vpull Q :: \Lambda'
        \tag{inversion on 1} \\
        & && 4\quad
        && \Gamma; \Delta, y:A \vpull P :: \Lambda ~\text{with only $\ast$ rules}
        \tag{IH on 2} \\
        & && 5\quad
        && \Gamma; \Delta', x:B \vpull Q :: \Lambda' ~\text{with only $\ast$ rules}
        \tag{IH on 3} \\
        & && 6\quad
        && \Gamma; \Delta \vpull P :: \Lambda, y:\dual{A}
        \tag{\ruleLabel{$\mright$} on 4} \\
        & && 7\quad
        && \Gamma; \Delta' \vpull Q :: \Lambda', x:\dual{B}
        \tag{\ruleLabel{$\mright$} on 5} \\
        & && 8\quad
        && \Gamma; \Delta, \Delta' \vpull \nu{y}\send{x}{y}.(P \| Q) :: \Lambda, \Lambda', x:\dual{A} \tensor \dual{B}
        \tag{\ruleLabel{$\tensor$R} on 6 and 7} \\
        & && 9\quad
        && \Gamma; \Delta, \Delta', x:A \parr B \vpull \nu{y}\send{x}{y}.(P \| Q) :: \Lambda, \Lambda'
        \tag{\ruleLabel{$\mleft$} on 8}
        \displaybreak[1] \\[5pt]
        & \bullet \ruleLabel{copyR}
        && 1\quad
        && \Gamma, u:A; \Delta \vpull \nu{x}\send{u}{x}.P :: \Lambda
        \tag{assumption} \\
        & && 2\quad
        && \Gamma, u:A; \Delta \vpull P :: \Lambda, x:\dual{A}
        \tag{inversion on 1} \\
        & && 3\quad
        && \Gamma, u:A; \Delta \vpull P :: \Lambda, x:\dual{A} ~\text{with only $\ast$ rules}
        \tag{IH on 2} \\
        & && 4\quad
        && \Gamma, u:A; \Delta, x:A \vpull P :: \Lambda
        \tag{\ruleLabel{$\mleft$} on 3} \\
        & && 5\quad
        && \Gamma, u:A; \Delta \vpull \nu{x}\send{u}{x}.P :: \Lambda
        \tag{\ruleLabel{copyL} on 4}
        \displaybreak[1] \\[5pt]
        & \bullet \ruleLabel{$\whynot$R}
        && 1\quad
        && \Gamma; \Delta \vpull P\subst{x/u} :: \Lambda, x:\whynot \dual{A}
        \tag{assumption} \\
        & && 2\quad
        && \Gamma, u:A; \Delta \vpull P :: \Lambda
        \tag{inversion on 1} \\
        & && 3\quad
        && \Gamma, u:A; \Delta \vpull P :: \Lambda ~\text{with only $\ast$ rules}
        \tag{IH on 2} \\
        & && 4\quad
        && \Gamma; \Delta, x:\bang A \vpull P\subst{x/u} :: \Lambda
        \tag{\ruleLabel{$\bang$L} on 3} \\
        & && 5\quad
        && \Gamma; \Delta \vpull P\subst{x/u} :: \Lambda, x:\whynot \dual{A}
        \tag{\ruleLabel{$\mright$} on 4}
        \displaybreak[1] \\[5pt]
        & \bullet \ruleLabel{$\whynot$L}
        && 1\quad
        && \Gamma; x:\whynot A \vpull \serv{x}{y}.P :: \emptyset
        \tag{assumption} \\
        & && 2\quad
        && \Gamma; y:A \vpull P :: \emptyset
        \tag{inversion on 1} \\
        & && 3\quad
        && \Gamma; y:A \vpull P :: \emptyset ~\text{with only $\ast$ rules}
        \tag{IH on 2} \\
        & && 4\quad
        && \Gamma; \emptyset \vpull P :: y:\dual{A}
        \tag{\ruleLabel{$\mright$} on 3} \\
        & && 5 \quad
        && \Gamma; \emptyset \vpull \serv{x}{y}.P :: x:\bang \dual{A}
        \tag{\ruleLabel{$\bang$R} on 4} \\
        & && 6 \quad
        && \Gamma; x:\whynot A \vpull \serv{x}{y}.P :: \emptyset
        \tag{\ruleLabel{$\mleft$} on 5}
        \displaybreak[1] \\[5pt]
        & \bullet \ruleLabel{cutRR}
        && 1\quad
        && \Gamma ; \Delta , \Delta' \vpull \nu{x} ( P \| Q ) :: \Lambda , \Lambda'
        \tag{assumption} \\
        & && 2\quad
        && \Gamma ; \Delta \vpull P :: \Lambda , x:A \\
        & && 3\quad
        && \Gamma ; \Delta' \vpull Q :: \Lambda' , x:\dual{A}
        \tag{inversion on 1} \\
        & && 4\quad
        && \Gamma ; \Delta \vpull P :: \Lambda , x:A ~\text{with only $\ast$ rules}
        \tag{IH on 2} \\
        & && 5\quad
        && \Gamma ; \Delta' \vpull Q :: \Lambda' , x:\dual{A} ~\text{with only $\ast$ rules}
        \tag{IH on 3} \\
        & && 6\quad
        && \Gamma ; \Delta' , x:A \vpull Q :: \Lambda'
        \tag{\ruleLabel{$\mleft$} on 5} \\
        & && 7\quad
        && \Gamma; \Delta, \Delta' \vpull \nu{x} ( P \| Q ) :: \Lambda, \Lambda'
        \tag{\ruleLabel{cutRL} on 4 and 6}
        \displaybreak[1] \\[5pt]
        & \bullet \ruleLabel{cutLL}
        && 1\quad
        && \Gamma; \Delta, \Delta' \vpull \nu{x}(P \| Q) :: \Lambda, \Lambda'
        \tag{assumption} \\
        & && 2\quad
        && \Gamma; \Delta, x:A \vpull P :: \Lambda \\
        & && 3\quad
        && \Gamma; \Delta', x:\dual{A} \vpull Q :: \Lambda'
        \tag{inversion on 1} \\
        & && 4\quad
        && \Gamma; \Delta, x:A \vpull P :: \Lambda ~\text{with only $\ast$ rules}
        \tag{IH on 2} \\
        & && 5\quad
        && \Gamma; \Delta', x:\dual{A} \vpull Q :: \Lambda' ~\text{with only $\ast$ rules}
        \tag{IH on 3} \\
        & && 6\quad
        && \Gamma; \Delta \vpull P :: \Lambda, x:\dual{A}
        \tag{\ruleLabel{$\mright$} on 4} \\
        & && 7\quad
        && \Gamma; \Delta, \Delta' \vpull \nu{x}(P \| Q) :: \Lambda, \Lambda'
        \tag{\ruleLabel{cutRL} on 6 and 5}
        \displaybreak[1] \\[5pt]
        & \bullet \ruleLabel{cut$\whynot$R}
        && 1\quad
        && \Gamma; \Delta \vpull \nu{u} ( P \| \serv{u}{x}.Q ) :: \Lambda
        \tag{assumption} \\
        & && 2\quad
        && \Gamma, u:A; \Delta \vpull P :: \Lambda \\
        & && 3\quad
        && \Gamma; x:\dual{A} \vpull Q :: \emptyset
        \tag{inversion on 1} \\
        & && 4\quad
        && \Gamma, u:A; \Delta \vpull P :: \Lambda ~\text{with only $\ast$ rules}
        \tag{IH on 2} \\
        & && 5\quad
        && \Gamma; x:\dual{A} \vpull Q :: \emptyset ~\text{with only $\ast$ rules}
        \tag{IH on 3} \\
        & && 6\quad
        && \Gamma; \emptyset \vpull Q :: x:A
        \tag{\ruleLabel{$\mright$} on 5} \\
        & && 7\quad
        && \Gamma; \Delta \vpull \nu{u} ( P \| \serv{u}{x}.Q ) :: \Lambda
        \tag{\ruleLabel{cut$\bang$R} on 4 and 6}
        \displaybreak[1] \\[5pt]
        & \bullet \ruleLabel{cut$\whynot$L}
        && 1\quad
        && \Gamma; \Delta \vpull \nu{u}(\serv{u}{x}.P \| Q) :: \Lambda
        \tag{assumption} \\
        & && 2\quad
        && \Gamma; x:\dual{A} \vpull P :: \emptyset \\
        & && 3\quad
        && \Gamma, u:A; \Delta \vpull Q :: \Lambda
        \tag{inversion on 1} \\
        & && 4\quad
        && \Gamma; x:\dual{A} \vpull P :: \emptyset ~\text{with only $\ast$ rules}
        \tag{IH on 2} \\
        & && 5\quad
        && \Gamma, u:A; \Delta \vpull Q :: \Lambda ~\text{with only $\ast$ rules}
        \tag{IH on 3} \\
        & && 6\quad
        && \Gamma; \emptyset \vpull P :: x:A
        \tag{\ruleLabel{$\mright$} on 4} \\
        & && 7\quad
        && \Gamma; \Delta \vpull \nu{u}(\serv{u}{x}.P \| Q) :: \Lambda
        \tag{\ruleLabel{cut$\bang$L} on 6 and 5}
    \end{align*}

    \emph{(Item 2)}
    Suppose given a proof of $\Gamma; \Delta \vdash P :: \Lambda$, possibly with applications of Rules~\ruleLabel{$\mleft$} and~\ruleLabel{$\mright$}.
    By applying induction on the structure of this proof we show that it can be transformed to not contain any applications of \ruleLabel{$\mleft$} and \ruleLabel{$\mright$}.
    We discuss every possible last-applied rule.
    However, all cases except \ruleLabel{$\mleft$} and \ruleLabel{$\mright$} follow directly from the induction hypothesis.
    Therefore, we only detail the case of \ruleLabel{$\mleft$}---the case of \ruleLabel{$\mright$} is analogous.

    Since \ruleLabel{$\mleft$} is the last-applied rule, we know the assumption is of the form $\Gamma; \Delta, x:\dual{A} \vdash P :: \Lambda$.
    By inversion, we have $\Gamma; \Delta \vdash P :: \Lambda, x:A$.
    The idea is to move the application of \ruleLabel{$\mleft$} up the proof tree, applied to a \beginbas subterm \endbas of $A$.
    We apply the induction hypothesis to find a proof of $\Gamma; \Delta \vdash P :: \Lambda, x:A$ without applications of \ruleLabel{$\mleft$} and \ruleLabel{$\mright$}.
    Typing rules leave all channels/types untouched except the ones they work on.
    Therefore, we can traverse up the proof tree---remembering which steps were taken---until we encounter the rule that introduces $x:A$.
    Note that these steps do not include applications of \ruleLabel{$\mleft$} and \ruleLabel{$\mright$}.
    The consequence of this rule looks like $\Gamma'; \Delta' \vdash P' :: \Lambda', x:A$, for some $\Gamma'$, $\Delta'$, $P'$ and $\Lambda'$.
    Now, we apply induction on the size of $A$ (with induction hypothesis denoted~\ih2) to prove $\Gamma'; \Delta', x:\dual{A} \vdash P' :: \Lambda'$.
    We discuss every possible last-applied rule that introduces~$x:A$:

    \begin{align*}
        & \bullet \ruleLabel{idR}
        && 1\quad
        && \Gamma'; y:A \vdash \fwd{y}{x} :: x:A
        \tag{assumption} \\
        & && 2\quad
        && \Gamma'; y:A, x:\dual{A} \vdash \fwd{y}{x} :: \emptyset
        \tag{\ruleLabel{idL}}
        \displaybreak[1] \\[5pt]
        & \bullet \ruleLabel{$\1$R}
        && 1\quad
        && \Gamma'; \emptyset \vdash \send{x}{}.\0 :: x:\1
        \tag{assumption} \\
        & && 2\quad
        && \Gamma'; x:\bot \vdash \send{x}{}.\0 :: \emptyset
        \tag{\ruleLabel{$\bot$L}}
        \displaybreak[1] \\[5pt]
        & \bullet \ruleLabel{$\bot$R}
        && 1\quad
        && \Gamma'; \Delta' \vdash \recv{x}{}.P' :: \Lambda', x:\bot ~\text{without \ruleLabel{$\mleft$}/\ruleLabel{$\mright$}}
        \tag{assumption} \\
        & && 2\quad
        && \Gamma'; \Delta' \vdash P' :: \Lambda'
        \tag{inversion on 1} \\
        & && 3\quad
        && \Gamma'; \Delta', x:\1 \vdash \recv{x}{}.P' :: \Lambda'
        \tag{\ruleLabel{$\1$L} on 2}
        \displaybreak[1] \\[5pt]
        & \bullet \ruleLabel{$\tensor$R}
        && 1\quad
        && \Gamma'; \Delta', \Delta'' \vdash \nu{y}\send{x}{y}.(P' \| Q') :: \Lambda', \Lambda'', x:B \tensor C \\
        & &&
        && ~\text{without \ruleLabel{$\mleft$}/\ruleLabel{$\mright$}}
        \tag{assumption} \\
        & && 2\quad
        && \Gamma'; \Delta' \vdash P' :: \Lambda', y:B \\
        & && 3\quad
        && \Gamma'; \Delta'' \vdash Q' :: \Lambda'', x:C
        \tag{inversion on 1} \\
        & && 4\quad
        && \Gamma'; \Delta', y:\dual{B} \vdash P' :: \Lambda'
        \tag{\ruleLabel{$\mleft$} on 2} \\
        & && 5\quad
        && \Gamma'; \Delta'', x:\dual{C} \vdash Q' :: \Lambda''
        \tag{\ruleLabel{$\mleft$} on 3} \\
        & && 6\quad
        && \Gamma'; \Delta', y:\dual{B} \vdash P' :: \Lambda' ~\text{without \ruleLabel{$\mleft$}/\ruleLabel{$\mright$}}
        \tag{\ih2 on 4} \\
        & && 7\quad
        && \Gamma'; \Delta'', x:\dual{C} \vdash Q' :: \Lambda'' ~\text{without \ruleLabel{$\mleft$}/\ruleLabel{$\mright$}}
        \tag{\ih2 on 5} \\
        & && 8\quad
        && \Gamma'; \Delta', \Delta'', x:\dual{B} \parr \dual{C} \vdash \nu{y}\send{x}{y}.(P' \| Q') :: \Lambda', \Lambda''
        \tag{\ruleLabel{$\parr$L} on 8 and 9}
        \displaybreak[1] \\[5pt]
        & \bullet \ruleLabel{$\lolli$R}
        && 1\quad
        && \Gamma'; \Delta' \vdash \recv{x}{y}.P' :: \Lambda', x:B \lolli C ~\text{without \ruleLabel{$\mleft$}/\ruleLabel{$\mright$}}
        \tag{assumption} \\
        & && 2\quad
        && \Gamma'; \Delta', y:B \vdash P' :: \Lambda', x:C
        \tag{inversion on 1} \\
        & && 3\quad
        && \Gamma'; \Delta', y:B, x:\dual{C} \vdash P' :: \Lambda'
        \tag{\ruleLabel{$\mleft$} on 2} \\
        & && 4\quad
        && \Gamma'; \Delta', y:B, x:\dual{C} \vdash P' :: \Lambda' ~\text{without \ruleLabel{$\mleft$}/\ruleLabel{$\mright$}}
        \tag{\ih2 on 3} \\
        & && 5\quad
        && \Gamma'; \Delta', x:B \tensor \dual{C} \vdash \recv{x}{y}.P' :; \Lambda'
        \tag{\ruleLabel{$\tensor$L} on 4}
        \displaybreak[1] \\[5pt]
        & \bullet \ruleLabel{$\oplus$R}
        && 1\quad
        && \Gamma'; \Delta' \vdash x \triangleleft j.P' :: \Lambda', x:\oplus\{i:A_i\}_{i \in I} ~\text{without \ruleLabel{$\mleft$}/\ruleLabel{$\mright$}}
        \tag{assumption} \\
        & && 2\quad
        && \Gamma'; \Delta' \vdash P' :: \Lambda', x:A_j \\
        & && 3\quad
        && j \in I
        \tag{inversion on 1} \\
        & && 4\quad
        && \Gamma'; \Delta', x:\dual{A_j} \vdash P' :: \Lambda'
        \tag{\ruleLabel{$\mleft$} on 3} \\
        & && 5\quad
        && \Gamma'; \Delta', x:\dual{A_j} \vdash P' :: \Lambda'
        \tag{\ih2 on 4} \\
        & && 6\quad
        && \Gamma'; \Delta', x:\&\{i:\dual{A_i}\}_{i \in I} \vdash x \triangleleft j.P' :: \Lambda'
        \tag{\ruleLabel{$\&$L} on 5 and 3}
        \displaybreak[1] \\[5pt]
        & \bullet \ruleLabel{$\&$R}
        && 1\quad
        && \Gamma'; \Delta' \vdash x \triangleright \{i:P'_i\}_{i \in I} :: \Lambda', x:\&\{i:A_i\}_{i \in I} ~\text{without \ruleLabel{$\mleft$}/\ruleLabel{$\mright$}}
        \tag{assumption} \\
        & && 2\quad
        && \forall i \in I.~ \Gamma'; \Delta' \vdash P'_i :: \Lambda', x:A_i
        \tag{inversion on 1} \\
        & && 3\quad
        && \forall i \in I.~ \Gamma'; \Delta', x:\dual{A_i} \vdash P'_i :: \Lambda'
        \tag{\ruleLabel{$\mleft$} on 2} \\
        & && 4\quad
        && \forall i \in I.~ \Gamma'; \Delta', x:\dual{A_i} \vdash P'_i :: \Lambda' ~\text{without \ruleLabel{$\mleft$}/\ruleLabel{$\mright$}}
        \tag{\ih2 on 3} \\
        & && 5\quad
        && \Gamma'; \Delta', x:\oplus\{i:\dual{A_i}\}_{i \in I} \vdash x \triangleright \{i:P'_i\}_{i \in I} :: \Lambda'
        \tag{\ruleLabel{$\oplus$L} on 4}
        \displaybreak[1] \\[5pt]
        & \bullet \ruleLabel{$\bang$R}
        && 1\quad
        && \Gamma'; \emptyset \vdash \serv{x}{y}.P' :: x:\bang B ~\text{without \ruleLabel{$\mleft$}/\ruleLabel{$\mright$}}
        \tag{assumption} \\
        & && 2\quad
        && \Gamma'; \emptyset \vdash P' :: y:B
        \tag{inversion on 1} \\
        & && 3\quad
        && \Gamma'; y:\dual{B} \vdash P' :: \emptyset
        \tag{\ruleLabel{$\mleft$} on 2} \\
        & && 4\quad
        && \Gamma'; y:\dual{B} \vdash P' :: \emptyset ~\text{without \ruleLabel{$\mleft$}/\ruleLabel{$\mright$}}
        \tag{\ih2 on 3} \\
        & && 5\quad
        && \Gamma'; x:\whynot \dual{B} \vdash \serv{x}{y}.P' :: \emptyset
        \tag{\ruleLabel{$\whynot$L} on 4}
        \displaybreak[1] \\[5pt]
        & \bullet \ruleLabel{$\whynot$R}
        && 1\quad
        && \Gamma'; \Delta' \vdash P'\subst{x/u} :: \Lambda', x:\whynot B
        \tag{assumption} \\
        & && 2\quad
        && \Gamma', u:\dual{A}; \Delta' \vdash P' :: \Lambda'
        \tag{inversion on 1} \\
        & && 3\quad
        && \Gamma'; \Delta', x:\bang \dual{B} \vdash P'\subst{x/u} :: \Lambda'
        \tag{\ruleLabel{$\bang$L} on 2}
    \end{align*}

    \noindent
    Finally, we recall the steps we have traversed up the tree and remember that they do not include applications of \ruleLabel{$\mleft$} and \ruleLabel{$\mright$}.
    We re-apply them on $\Gamma'; \Delta', x:\dual{A} \vdash P' :: \Lambda'$,
    without affecting~$x:\dual{A}$.
    Instead, they only affect $\Gamma'$, $\Delta'$ and $\Lambda'$ to give us a proof of $\Gamma; \Delta, x:\dual{A} \vdash P :: \Lambda$ without applications of \ruleLabel{$\mleft$} and \ruleLabel{$\mright$}.

This concludes the proof of \Cref{thm:madmissible}.
\end{proof}


\section{Comparing Intuitionistic and Classical Interpretations}
\label{sec:comparison}

In this section, we rigorously compare the class of \pull-typable processes to the classes of processes typable in session type interpretations of linear logic, in classical~\cite{journal/mscs/CairesPT16,conf/icfp/Wadler12} and intuitionistic~\cite{conf/concur/CairesP10} settings.
\pull is an independent yardstick for this comparison, because it is derived from \ull which subsumes both linear logics by design.

We have discussed in \Cref{sec:ull} several design choices we made when defining \ull and \pull.
Besides differences stemming from the dichotomy between classical and intuitionistic linear logic,  logically-motivated session type systems also present features induced by certain design choices.
For a fair comparison, we want to make sure that the differences come only from typing.
This means that we need to make the same design choices for both interpretations: we require explicit closing, a separate unrestricted context, and identity as forwarding.

\subsection{Session Type Systems Derived from Intuitionistic and Classical Linear Logic}
\label{ssec:types}

\begin{figure}[p]
    \begin{mathpar}
        \begin{bussproof}[id]
            \bussAx{
                \Gamma; x:A \vdill \fwd{x}{y} :: y:A
            }
        \end{bussproof}
        \and
        \begin{bussproof}[$\1$R]
            \bussAx{
                \Gamma; \emptyset \vdill \send{x}{}.\0 :: x:\1
            }
        \end{bussproof}
        \and
        \begin{bussproof}[$\1$L]
            \bussAssume{
                \Gamma; \Delta \vdill P :: z:C
            }
            \bussUn{
                \Gamma; \Delta, x:\1 \vdill \recv{x}{}.P :: z:\beginbas{C}\endbas
            }
        \end{bussproof}
        \and
        \begin{bussproof}[$\tensor$R]
            \bussAssume{
                \Gamma; \Delta \vdill P :: y:A
            }
            \bussAssume{
                \Gamma; \Delta' \vdill Q :: x:B
            }
            \bussBin{
                \Gamma; \Delta, \Delta' \vdill \nu{y}\send{x}{y}.(P \| Q) :: x:A \tensor B
            }
        \end{bussproof}
        \and
        \begin{bussproof}[$\tensor$L]
            \bussAssume{
                \Gamma; \Delta, y:A, x:B \vdill P :: z:C
            }
            \bussUn{
                \Gamma; \Delta, x:A \tensor B \vdill \recv{x}{y}.P :: z:C
            }
        \end{bussproof}
        \and
        \begin{bussproof}[$\lolli$R]
            \bussAssume{
                \Gamma; \Delta, y:A \vdill P :: x:B
            }
            \bussUn{
                \Gamma; \Delta \vdill \recv{x}{y}.P :: x:A \lolli B
            }
        \end{bussproof}
        \and
        \begin{bussproof}[$\lolli$L]
            \bussAssume{
                \Gamma; \Delta \vdill P :: y:A
            }
            \bussAssume{
                \Gamma; \Delta', x:B \vdill Q :: z:C
            }
            \bussBin{
                \Gamma; \Delta, \Delta', x:A \lolli B \vdill \nu{y}\send{x}{y}.(P \| Q) :: z:C
            }
        \end{bussproof}
        \and
        \begin{bussproof}[$\oplus$R]
            \bussAssume{
                \Gamma; \Delta \vdill P :: x:A_j
            }
            \bussAssume{
                j \in I
            }
            \bussBin{
                \Gamma; \Delta \vdill x \triangleleft j .P :: x:\oplus\{i:A_i\}_{i \in I}
            }
        \end{bussproof}
        \and
        \begin{bussproof}[$\oplus$L]
            \bussAssume{
                \forall i \in I.~ \Gamma; \Delta, x:A_i \vdill P_i :: z:C
            }
            \bussUn{
                \Gamma; \Delta, x:\oplus\{i:A_i\}_{i \in I} \vdill x \triangleright \{i:P_i\}_{i \in I} :: z:C
            }
        \end{bussproof}
        \and
        \begin{bussproof}[$\&$R]
            \bussAssume{
                \forall i \in I.~ \Gamma; \Delta \vdill P_i :: x:A_i
            }
            \bussUn{
                \Gamma; \Delta \vdill x \triangleright \{i:P_i\}_{i \in I} :: x:\&\{i:A_i\}_{i \in I}
            }
        \end{bussproof}
        \and
        \begin{bussproof}[$\&$L]
            \bussAssume{
                \Gamma; \Delta, x:A_j \vdill P :: z:C
            }
            \bussAssume{
                j \in I
            }
            \bussBin{
                \Gamma; \Delta, x:\&\{i:A_i\}_{i \in I} \vdill x \triangleleft j .P :: z:C
            }
        \end{bussproof}
        \and
        \begin{bussproof}[copy]
            \bussAssume{
                \Gamma, u:A; \Delta, x:A \vdill P :: z:C
            }
            \bussUn{
                \Gamma, u:A; \Delta \vdill \nu{x}\send{u}{x}.P :: z:C
            }
        \end{bussproof}
        \and
        \begin{bussproof}[$\bang$R]
            \bussAssume{
                \Gamma; \emptyset \vdill P :: y:A
            }
            \bussUn{
                \Gamma; \emptyset \vdill \serv{x}{y}.P :: x:\bang A
            }
        \end{bussproof}
        \and
        \begin{bussproof}[$\bang$L]
            \bussAssume{
                \Gamma, u:A; \Delta \vdill P :: z:C
            }
            \bussUn{
                \Gamma; \Delta, x:\bang A \vdill P\subst{x/u} :: z:C
            }
        \end{bussproof}
        \and
        \begin{bussproof}[cutRL]
            \bussAssume{
                \Gamma; \Delta \vdill P :: x:A
            }
            \bussAssume{
                \Gamma; \Delta', x:A \vdill Q :: z:C
            }
            \bussBin{
                \Gamma; \Delta, \Delta' \vdill \nu{x}(P \| Q) :: z:C
            }
        \end{bussproof}
        \and
        \begin{bussproof}[cutLR]
            \bussAssume{
                \Gamma; \Delta , x:A \vdill P :: z:C
            }
            \bussAssume{
                \Gamma; \Delta' \vdill Q :: x:A
            }
            \bussBin{
                \Gamma; \Delta, \Delta' \vdill \nu{x}(P \| Q) :: z:C
            }
        \end{bussproof}
        \and
        \begin{bussproof}[cut$\bang$R]
            \bussAssume{
                \Gamma, u:A; \Delta \vdill P :: z:C
            }
            \bussAssume{
                \Gamma; \emptyset \vdill Q :: x:A
            }
            \bussBin{
                \Gamma; \Delta \vdill \nu{u} ( P \| \serv{u}{x}.Q ) :: z:C
            }
        \end{bussproof}
        \and
        \begin{bussproof}[cut$\bang$L]
            \bussAssume{
                \Gamma; \emptyset \vdill P :: x:A
            }
            \bussAssume{
                \Gamma, u:A; \Delta \vdill Q :: z:C
            }
            \bussBin{
                \Gamma; \Delta \vdill \nu{u}(\serv{u}{x}.P \| Q) :: z:C
            }
        \end{bussproof}
    \end{mathpar}

    \caption{The \pill type system.}\label{f:ill-inf}
\end{figure}

\begin{figure}[t!]
    \begin{mathpar}
        \begin{bussproof}[id]
            \bussAx{
                \fwd{x}{y} \vdcll \Gamma; x:A, y:\dual{A}
            }
        \end{bussproof}
        \and
        \begin{bussproof}[$\bot$]
            \bussAssume{
                P \vdcll \Gamma; \Delta
            }
            \bussUn{
                \recv{x}{}.P \vdcll \Gamma; \Delta, x:\bot
            }
        \end{bussproof}
        \and
        \begin{bussproof}[$\1$]
            \bussAx{
                \send{x}{}.\0 \vdcll \Gamma; x:\1
            }
        \end{bussproof}
        \and
        \begin{bussproof}[$\tensor$]
            \bussAssume{
                P \vdcll \Gamma; \Delta, y:A
            }
            \bussAssume{
                Q \vdcll \Gamma; \Delta', x:B
            }
            \bussBin{
                \nu{y}\send{x}{y}.(P \| Q) \vdcll \Gamma; \Delta, \Delta', x:A \tensor B
            }
        \end{bussproof}
        \and
        \begin{bussproof}[$\parr$]
            \bussAssume{
                P \vdcll \Gamma; \Delta, y:A, x:B
            }
            \bussUn{
                \recv{x}{y}.P \vdcll \Gamma; \Delta, x:A \parr B
            }
        \end{bussproof}
        \and
        \begin{bussproof}[$\oplus$]
            \bussAssume{
                P \vdcll \Gamma; \Delta, x:A_j
            }
            \bussAssume{
                j \in I
            }
            \bussBin{
                x \triangleleft j.P \vdcll \Gamma; \Delta, x:\oplus\{i:A_i\}_{i \in I}
            }
        \end{bussproof}
        \and
        \begin{bussproof}[$\&$]
            \bussAssume{
                \forall i \in I.~ P_i \vdcll \Gamma; \Delta, x:A_i
            }
            \bussUn{
                x \triangleright \{i:P_i\}_{i \in I} \vdcll \Gamma; \Delta, x:\&\{i:A_i\}_{i \in I}
            }
        \end{bussproof}
        \and
        \begin{bussproof}[copy]
            \bussAssume{
                P \vdcll \Gamma, u:A; \Delta, y:A
            }
            \bussUn{
                \nu{y}\send{u}{y}.P \vdcll \Gamma, u:A; \Delta
            }
        \end{bussproof}
        \and
        \begin{bussproof}[$\whynot$]
            \bussAssume{
                P \vdcll \Gamma, u:A; \Delta
            }
            \bussUn{
                P\subst{x/u} \vdcll \Gamma; \Delta, x:\whynot A
            }
        \end{bussproof}
        \and
        \begin{bussproof}[$\bang$]
            \bussAssume{
                P \vdcll \Gamma; y:A
            }
            \bussUn{
                \serv{x}{y}.P \vdcll \Gamma; x:\bang A
            }
        \end{bussproof}
        \and
        \begin{bussproof}[cut]
            \bussAssume{
                P \vdcll \Gamma; \Delta, x:A
            }
            \bussAssume{
                Q \vdcll \Gamma; \Delta', x:\dual{A}
            }
            \bussBin{
                \nu{x}(P \| Q) \vdcll \Gamma; \Delta, \Delta'
            }
        \end{bussproof}
        \and
        \begin{bussproof}[cut$\whynot$R]
            \bussAssume{
                P \vdcll \Gamma, u:A; \Delta
            }
            \bussAssume{
                Q \vdcll \Gamma; x:\dual{A}
            }
            \bussBin{
                \nu{u} ( P \| \serv{u}{x}.Q ) \vdcll \Gamma; \Delta
            }
        \end{bussproof}
        \and
        \begin{bussproof}[cut$\whynot$L]
            \bussAssume{
                P \vdcll \Gamma; x:\dual{A}
            }
            \bussAssume{
                Q \vdcll \Gamma, u:A; \Delta
            }
            \bussBin{
                \nu{u}(\serv{u}{x}.P \| Q) \vdcll \Gamma; \Delta
            }
        \end{bussproof}
    \end{mathpar}

    \caption{The \pcllcp type system.}\label{f:cll-inf}
\end{figure}

\noindent
Our goal is to derive session type systems from intuitionistic and classical linear logic for the same process syntax as \pull uses (cf.\ \Cref{def:processterms}).

\paragraph{The Intuitionistic Type System (\pill)}
\Cref{f:ill-inf} gives the inference rules for the type system derived from \emph{intuitionistic} linear logic, denoted \pill.
Following the presentation by Caires, Pfenning, and Toninho in~\cite{conf/tldi/CairesPT12},
propositions/types for the intuitionistic type system are generated by the following grammar (a fragment of \Cref{def:propositions}):
    \begin{align*}
        A, B ::= \1   \bnf A \tensor B \bnf A \lolli B \bnf \oplus\{i:A\}_{i \in I} \bnf \&\{i:A\}_{i \in I} \bnf \bang A
    \end{align*}
\beginjp Notice how, following most presentations of intuitionistic linear logic (e.g.,~\cite{DBLP:conf/tapsoft/GirardL87,report/BarberP96,conf/lics/L18}),  \pill propositions do not include $\bot$ and $?$. \endjp
The judgment is denoted as follows:
$$
\Gamma; \Delta \vdill P :: z:C
$$
With respect to the intuitionistic interpretation introduced in~\cite{conf/concur/CairesP10,journal/mscs/CairesPT16}, \pill features a non-silent interpretation of $\1$, based on explicit closure of sessions (cf.\ Rules~\ruleLabel{$\1$R} and~\ruleLabel{$\1$L}).
We adopt this interpretation because, as explained in~\cite{conf/tldi/CairesPT12}, it leads to a Curry-Howard correspondence that is tighter than correspondences with silent interpretations (such as those in~\cite{conf/concur/CairesP10,journal/mscs/CairesPT16}).

A more superficial difference is that \pill follows standard presentations of session type systems by supporting $n$-ary labeled choices; in contrast, the systems in~\cite{conf/concur/CairesP10,journal/mscs/CairesPT16} support binary labeled choices.
{We also include symmetric variants of the cut-rules, as they are necessary for type preservation under Rule~\ruleLabel{CutSymm} of structural congruence.}
\beginjp As expected, none of these cut-rules  make use of the duality relation in \Cref{def:duality}.\endjp

\paragraph{The Classical Type System (\pcllcp)}
\Cref{f:cll-inf} gives the inference rules for the type system derived from \emph{classical} linear logic.
It is based on a combination of features from Caires, Pfenning, and Toninho's \pcll in~\cite{journal/mscs/CairesPT16} and Wadler's \cp in~\cite{conf/icfp/Wadler12};
in the following, it is denoted \pcllcp.
The syntax of propositions/types is exactly as in \Cref{def:propositions}.
The corresponding judgment is as follows:
$$
P \vdcll \Gamma; \Delta
$$
\Cref{tbl:pcll_cp_pcllcp} summarizes the differences in the design choices between \pcll, \cp, and \pcllcp; these differences are merely superficial:
\begin{itemize}[noitemsep]
    \item
        As we have seen, an explicit closing of sessions (as in \cp) concerns a non-silent interpretation of the atomic propositions $\1$ and $\bot$.
        In contrast, \pcll realizes an implicit (silent) closing of sessions.

    \item
        Sequents with a separate unrestricted context (as in \pcll) are of the form $P \vdash \Gamma; \Delta$, which can also be written as $P \vdash \Delta, \Gamma'$ where $\Gamma'$ contains only types of the form $!A$.

    \item
        The identity axiom can be interpreted as the forwarding process, which enables \cp to account for behavioral polymorphism (i.e., universal and existential quantification over propositions/session types)~\cite{conf/icfp/Wadler12,conf/esop/CairesPPT13}.
        As already mentioned, forwarding is not a typical process construct in session $\pi$-calculi.
\end{itemize}
 {Note that, since \pcllcp typing judgments are one-sided, there is no need for symmetric cut-rules.} \beginjp The three cut-rules in \Cref{f:cll-inf} use duality  as in \Cref{def:duality}.\endjp

\subsection{Formal Comparison}

\begin{table}[]
\newcolumntype{Y}{>{\centering\arraybackslash}X}
\centering
\begin{tabularx}{.85\textwidth}{lYYY}
 & \textbf{Explicit closing}
 & \textbf{Separate \mbox{unrestricted} context}
 & \textbf{Identity as \mbox{forwarding}} \\ \hline
\pcll~\cite{journal/mscs/CairesPT16}   & No  & Yes & No  \\
\cp~\cite{conf/icfp/Wadler12}     & Yes & No  & Yes \\
\pcllcp (this paper) & Yes & Yes & Yes
\end{tabularx}
\caption{Feature comparison of three session type interpretations of classical linear logic.}\label{tbl:pcll_cp_pcllcp}
\end{table}

Now that we have presented all three systems, we start our comparison by contrasting the shape of their typing judgments:

\begin{minipage}[T]{.28\textwidth}
    \[
        \Gamma; \Delta \vdash P :: \Lambda
    \]
\end{minipage}%
\begin{minipage}[T]{.40\textwidth}
    \[
        \Gamma; \Delta \vdill P :: x:A
    \]
\end{minipage}%
\begin{minipage}[T]{.28\textwidth}
    \[
        P \vdcll \Gamma; \Delta
    \]
\end{minipage}%
\hfill

\medskip
\noindent
In \pull and \pill judgments are similar, but in \pill they have exactly one channel/type pair on the right.
We will see that the difference between \pull and \pill can be characterized by this fact alone.
Judgments in \pcllcp are different from those in \pull and \pill:
they have only one linear context and both the linear and the unrestricted contexts are on the right.
As we will see, our results reflect this with a duality relation between the contexts of \pull and \pcllcp.

Our formal results rely on classes of processes typable in the three typing systems:

\begin{definition}\label{def:typablesets}
    Let $\mathbb{P}$ denote the set of all processes induced by \Cref{def:processterms}.
    Then
    \begin{align*}
        \U &= \{P \in \mathbb{P} \mid \exists \Gamma, \Delta, \Lambda ~\text{such that}~ \Gamma; \Delta \vdash P :: \Lambda\}, \\
        \C &= \{P \in \mathbb{P} \mid \exists \Gamma, \Delta ~\text{such that}~ P \vdcll \Gamma; \Delta\}, \\
        \I &= \{P \in \mathbb{P} \mid \exists \Gamma, \Delta, x, A ~\text{such that}~ \Gamma; \Delta \vdill P :: x:A\}.
    \end{align*}
\end{definition}

Our first result is that $\U = \C$, i.e.,
\pull is merely a two-sided representation of \pcllcp.

\begin{theorem}\label{thm:cllull}
    $\U = \C$.
\end{theorem}

We briefly discuss how we prove this result.
On the one hand (from left to right), if $P \in \U$, there is a proof of $\Gamma; \Delta \vdash P :: \Lambda$.
By backtracking on this proof, we can use rules in \pcllcp analogous to the \pull-rules used to generate an equivalent proof of $P \vdcll \dual{(\Gamma)}; \dual{(\Delta)}, \Lambda$ thus showing $P \in \C$.
Note how the single-sidedness of \pcllcp judgments requires us to move $\Gamma$ and $\Delta$ to the right-hand side  using duality.

On the other hand (from right to left), if $P \in \C$, there is a proof of $P \vdcll \Gamma; \Delta$.
Again, by backtracking on this proof, we can use a combination of rules similar to those used in \pcllcp in combination with Rules~\ruleLabel{$\mleft$} and~\ruleLabel{$\mright$} from \Cref{def:pullm} to prove $\dual{(\Gamma)}; \emptyset \vpull \Delta$ in \pullm.
Going through \pullm simplifies this process, since using \ruleLabel{$\mleft$} and \ruleLabel{$\mright$} enables us to guarantee that all of $\Delta$ ends up on the right-hand side of the typing judgment instead of dividing unpredictably between left and right.
Note that we do have to use duality to move $\Gamma$ to the left.
Since \ruleLabel{$\mleft$} and \ruleLabel{$\mright$} are admissible in \pull (cf.\ \Cref{thm:madmissible}),
this means we also have a proof of $\dual{(\Gamma)}; \emptyset \vdash \Delta$ in \pull.
Hence, $P \in \U$.

\begin{proof}[Proof of \Cref{thm:cllull}]
    \emph{($\,\U \subseteq \C$)}
    Take any $P \in \U$.
    Then, by \Cref{def:typablesets}, there are $\Gamma$, $\Delta$, $\Lambda$ s.t.\ $\Gamma; \Delta \vdash P :: \Lambda$.
    By showing that this implies $P \vdcll \dual{(\Gamma)}; \dual{(\Delta)}, \Lambda$, we have $P \in \C$.
    We show this by induction on the structure of the proof of $\Gamma; \Delta \vdash P :: \Lambda$.
    \begin{align*}
        & \bullet \ruleLabel{$\id$R}
        && 1\quad
        && \Gamma; x:A \vdash \fwd{x}{y} :: y:A
        \tag{assumption} \\
        & && 2\quad
        && \fwd{x}{y} \vdcll \dual{(\Gamma)}; x:\dual{A}, y:A
        \tag{\ruleLabel{id}}
        \displaybreak[1] \\[5pt]
        & \bullet \ruleLabel{idL}
        && 1\quad
        && \Gamma; x:A, y:\dual{A} \vdash \fwd{x}{y} :: \emptyset
        \tag{assumption} \\
        & && 2\quad
        && \fwd{x}{y} \vdcll \dual{(\Gamma)}; x:\dual{A}, y:A
        \tag{\ruleLabel{id}}
        \displaybreak[1] \\[5pt]
        & \bullet \ruleLabel{$\1$R}
        && 1\quad
        && \Gamma; \emptyset \vdash \send{x}{}.\0 :: x:\1
        \tag{assumption} \\
        & && 2\quad
        && \send{x}{}.\0 \vdcll \dual{(\Gamma)}; x:\1
        \tag{\ruleLabel{$\1$}}
        \displaybreak[1] \\[5pt]
        & \bullet \ruleLabel{$\1$L}
        && 1\quad
        && \Gamma; \Delta, x:\1 \vdash \recv{x}{}.P :: \Lambda
        \tag{assumption} \\
        & && 2\quad
        && \Gamma; \Delta \vdash P :: \Lambda
        \tag{inversion on 1} \\
        & && 3\quad
        && P \vdcll \dual{(\Gamma)}; \dual{(\Delta)}, \Lambda
        \tag{IH on 2} \\
        & && 4\quad
        && \recv{x}{}.P \vdcll \dual{(\Gamma)}; \dual{(\Delta)}, \Lambda, x:\bot
        \tag{\ruleLabel{$\bot$} on 3}
        \displaybreak[1] \\[5pt]
        & \bullet \ruleLabel{$\bot$R}
        && 1\quad
        && \Gamma; \Delta \vdash \recv{x}{}.P :: \Lambda, x:\bot
        \tag{assumption} \\
        & && 2\quad
        && \Gamma; \Delta \vdash P :: \Lambda
        \tag{inversion on 1} \\
        & && 3\quad
        && P \vdcll \dual{(\Gamma)}; \dual{(\Delta)}, \Lambda
        \tag{IH on 2} \\
        & && 4\quad
        && \recv{x}{}.P \vdcll \dual{(\Gamma)}; \dual{(\Delta)}, \Lambda, x:\bot
        \tag{\ruleLabel{$\bot$} on 3}
        \displaybreak[1] \\[5pt]
        & \bullet \ruleLabel{$\bot$L}
        && 1\quad
        && \Gamma; x:\bot \vdash \send{x}{}.\0 :: \emptyset
        \tag{assumption} \\
        & && 2\quad
        && \send{x}{}.\0 \vdcll \dual{(\Gamma)}; x:\1
        \tag{\ruleLabel{$\1$}}
        \displaybreak[1] \\[5pt]
        & \bullet \ruleLabel{$\tensor$R}
        && 1\quad
        && \Gamma; \Delta, \Delta' \vdash \nu{y}\send{x}{y}.(P \| Q) :: \Lambda, \Lambda', x:A \tensor B
        \tag{assumption} \\
        & && 2\quad
        && \Gamma; \Delta \vdash P :: \Lambda, y:A \\
        & && 3\quad
        && \Gamma; \Delta' \vdash Q :: \Lambda', x:B
        \tag{inversion on 1} \\
        & && 4\quad
        && P \vdcll \dual{(\Gamma)}; \dual{(\Delta)}, \Lambda, y:A
        \tag{IH on 2} \\
        & && 5\quad
        && Q \vdcll \dual{(\Gamma)}; \dual{(\Delta')}, \Lambda', x:B
        \tag{IH on 3} \\
        & && 6\quad
        && \nu{y}\send{x}{y}.(P \| Q) \vdcll \dual{(\Gamma)}; \dual{(\Delta)}, \dual{(\Delta')}, \Lambda, \Lambda', x:A \tensor B
        \tag{\ruleLabel{$\tensor$} on 4 and 5}
        \displaybreak[1] \\[5pt]
        & \bullet \ruleLabel{$\tensor$L}
        && 1\quad
        && \Gamma; \Delta, x:A \tensor B \vdash \recv{x}{y}.P :: \Lambda
        \tag{assumption} \\
        & && 2\quad
        && \Gamma; \Delta, y:A, x:B \vdash P :: \Lambda
        \tag{inversion on 1} \\
        & && 3\quad
        && P \vdcll \dual{(\Gamma)}; \dual{(\Delta)}, \Lambda, y:\dual{A}, x:\dual{B}
        \tag{IH on 2} \\
        & && 4\quad
        && \recv{x}{y}.P \vdcll \dual{(\Gamma)}; \dual{(\Delta)}, \Lambda, x:\dual{A} \parr \dual{B}
        \tag{\ruleLabel{$\parr$} on 3}
        \displaybreak[1] \\[5pt]
        & \bullet \ruleLabel{$\parr$R}
        && 1\quad
        && \Gamma; \Delta \vdash \recv{x}{y}.P :: \Lambda, x:A \parr B
        \tag{assumption} \\
        & && 2\quad
        && \Gamma; \Delta \vdash P :: \Lambda, y:A, x:B
        \tag{inversion on 1} \\
        & && 3\quad
        && P \vdcll \dual{(\Gamma)}; \dual{(\Delta)}, \Lambda, y:A, x:B
        \tag{IH on 2} \\
        & && 4\quad
        && \recv{x}{y}.P \vdcll \dual{(\Gamma)}; \dual{(\Delta)}, \Lambda, x:A \parr B
        \tag{\ruleLabel{$\parr$} on 3}
        \displaybreak[1] \\[5pt]
        & \bullet \ruleLabel{$\parr$L}
        && 1\quad
        && \Gamma; \Delta, \Delta', x:A \parr B \vdash \nu{y}\send{x}{y}.(P \| Q) :: \Lambda, \Lambda'
        \tag{assumption} \\
        & && 2\quad
        && \Gamma; \Delta, y:A \vdash P :: \Lambda \\
        & && 3\quad
        && \Gamma; \Delta', x:B \vdash Q :: \Lambda'
        \tag{inversion on 1} \\
        & && 4\quad
        && P \vdcll \dual{(\Gamma)}; \dual{(\Delta)}, \Lambda, y:\dual{A}
        \tag{IH on 2} \\
        & && 5\quad
        && Q \vdcll \dual{(\Gamma)}; \dual{(\Delta')}, \Lambda', x:\dual{B}
        \tag{IH on 3} \\
        & && 6\quad
        && \nu{y}\send{x}{y}.(P \| Q) \vdcll \dual{(\Gamma)}; \dual{(\Delta)}, \dual{(\Delta')}, \Lambda, \Lambda', x:\dual{A} \tensor \dual{B}
        \tag{\ruleLabel{$\tensor$} on 4 and 5}
        \displaybreak[1] \\[5pt]
        & \bullet \ruleLabel{$\lolli$R}
        && 1\quad
        && \Gamma; \Delta \vdash \recv{x}{y}.P :: \Lambda, x:A \lolli B
        \tag{assumption} \\
        & && 2\quad
        && \Gamma; \Delta, y:A \vdash P :: \Lambda, x:B
        \tag{inversion on 1} \\
        & && 3\quad
        && P \vdcll \dual{(\Gamma)}; \dual{(\Delta)}, \Lambda, y:\dual{A}, x:B
        \tag{IH on 2} \\
        & && 4\quad
        && \recv{x}{y}.P \vdcll \dual{(\Gamma)}; \dual{(\Delta)}, \Lambda, x:\dual{A} \parr B
        \tag{\ruleLabel{$\parr$} on 3}
        \displaybreak[1] \\[5pt]
        & \bullet \ruleLabel{$\lolli$L}
        && 1\quad
        && \Gamma; \Delta, \Delta', x:A \lolli B \vdash \nu{y}\send{x}{y}.(P \| Q) :: \Lambda, \Lambda'
        \tag{assumption} \\
        & && 2\quad
        && \Gamma; \Delta \vdash P :: \Lambda, y:A \\
        & && 3\quad
        && \Gamma; \Delta', x:B \vdash Q :: \Lambda'
        \tag{inversion on 1} \\
        & && 4\quad
        && P \vdcll \dual{(\Gamma)}; \dual{(\Delta)}, \Lambda, y:A
        \tag{IH on 2} \\
        & && 5\quad
        && Q \vdcll \dual{(\Gamma)}; \dual{(\Delta')}, \Lambda', x:\dual{B}
        \tag{IH on 3} \\
        & && 6\quad
        && \nu{y}\send{x}{y}.(P \| Q) \vdcll \dual{(\Gamma)}; \dual{(\Delta)}, \dual{(\Delta')}, \Lambda, \Lambda', x:A \tensor \dual{B}
        \tag{\ruleLabel{$\tensor$} on 4 and 5}
        \displaybreak[1] \\[5pt]
        & \bullet \ruleLabel{$\&$R}
        && 1\quad
        && \Gamma; \Delta \vdash x \triangleright \{i:P_i\}_{i \in I} :: \Lambda, x:\&\{i:A_i\}_{i \in I}
        \tag{assumption} \\
        & && 2\quad
        && \forall i \in I.~ \Gamma; \Delta \vdash P_i :: \Lambda, x:A_i
        \tag{inversion on 1} \\
        & && 3\quad
        && \forall i \in I.~ P_i \vdcll \dual{(\Gamma)}; \dual{(\Delta)}, \Lambda, x:A_i
        \tag{IH on 2} \\
        & && 4\quad
        && x \triangleright \{i:P_i\}_{i \in I} \vdcll \dual{(\Gamma)}; \dual{(\Delta)}, \Lambda, x:\&\{i:A_i\}_{i \in I}
        \tag{\ruleLabel{$\&$} on 3}
        \displaybreak[1] \\[5pt]
        & \bullet \ruleLabel{$\&$L}
        && 1\quad
        && \Gamma; \Delta, x:\&\{i:A_i\} \vdash x \triangleleft j.P :: \Lambda
        \tag{assumption} \\
        & && 2\quad
        && \Gamma; \Delta, x:A_j \vdash P :: \Lambda \\
        & && 3\quad
        && j \in I
        \tag{inversion on 1} \\
        & && 4\quad
        && P \vdcll \dual{(\Gamma)}; \dual{(\Delta)}, \Lambda, x:\dual{A_j}
        \tag{IH on 2} \\
        & && 5\quad
        && x \triangleright j.P \vdcll \dual{(\Gamma)}; \dual{(\Delta)}, \Lambda, x:\oplus\{i:\dual{A_i}\}_{i \in I}
        \tag{\ruleLabel{$\oplus$} on 4 and 3}
        \displaybreak[1] \\[5pt]
        & \bullet \ruleLabel{$\oplus$R}
        && 1\quad
        && \Gamma; \Delta \vdash x \triangleleft j.P :: \Lambda, x:\oplus\{i:A_i\}_{i \in I}
        \tag{assumption} \\
        & && 2\quad
        && \Gamma; \Delta \vdash P :: \Lambda, x:A_j \\
        & && 3\quad
        && j \in I
        \tag{inversion on 1} \\
        & && 4\quad
        && P \vdcll \dual{(\Gamma)}; \dual{(\Delta)}, \Lambda, x:A_j
        \tag{IH on 2} \\
        & && 5\quad
        && x \triangleleft j.P \vdcll \dual{(\Gamma)}; \dual{(\Delta)}, \Lambda, x:\oplus\{i:A_i\}_{i \in I}
        \tag{\ruleLabel{$\oplus$} on 4 and 3}
        \displaybreak[1] \\[5pt]
        & \bullet \ruleLabel{$\oplus$L}
        && 1\quad
        && \Gamma; \Delta, x:\oplus\{i:A_i\}_{i \in I} \vdash x \triangleright \{i:P_i\}_{i \in I} :: \Lambda
        \tag{assumption} \\
        & && 2\quad
        && \forall i \in I.~ \Gamma; \Delta, x:A_i \vdash P_i :: \Lambda
        \tag{inversion on 1} \\
        & && 3\quad
        && \forall i \in I.~ P_i \vdcll \dual{(\Gamma)}; \dual{(\Delta)}, \Lambda, x:\dual{A_i}
        \tag{IH on 2} \\
        & && 4\quad
        && x \triangleright \{i:P_i\}_{i \in I} \vdcll \dual{(\Gamma)}; \dual{(\Delta)}, \Lambda, x:\&\{i:\dual{A_i}\}_{i \in I}
        \tag{\ruleLabel{$\&$} on 4}
        \displaybreak[1] \\[5pt]
        & \bullet \ruleLabel{copyR}
        && 1\quad
        && \Gamma, u:A; \Delta \vdash \nu{x}\send{u}{x}.P :: \Lambda
        \tag{assumption} \\
        & && 2\quad
        && \Gamma, u:A; \Delta \vdash P :: \Lambda, x:\dual{A}
        \tag{inversion on 1} \\
        & && 3\quad
        && P \vdcll \dual{(\Gamma)}, u:\dual{A}; \dual{(\Delta)}, \Lambda, x:\dual{A}
        \tag{IH on 2} \\
        & && 4\quad
        && \nu{x}\send{u}{x}.P \vdcll \dual{(\Gamma)}, u:\dual{A}; \dual{(\Delta}), \Lambda
        \tag{\ruleLabel{copy} on 3}
        \displaybreak[1] \\[5pt]
        & \bullet \ruleLabel{copyL}
        && 1\quad
        && \Gamma, u:A; \Delta \vdash \nu{x}\send{u}{x}.P :: \Lambda
        \tag{assumption} \\
        & && 2\quad
        && \Gamma, u:A; \Delta, x:A \vdash P :: \Lambda
        \tag{inversion on 1} \\
        & && 3\quad
        && P \vdcll \dual{(\Gamma)}, u:\dual{A}; \dual{(\Delta)}, \Lambda, x:\dual{A}
        \tag{IH on 2} \\
        & && 4\quad
        && \nu{x}\send{u}{x}.P :: \dual{(\Gamma)}, u:\dual{A}; \dual{(\Delta)}, \Lambda
        \tag{\ruleLabel{copy} on 3}
        \displaybreak[1] \\[5pt]
        & \bullet \ruleLabel{$\bang$R}
        && 1\quad
        && \Gamma; \emptyset \vdash \serv{x}{y}.P :: x:\bang A
        \tag{assumption} \\
        & && 2\quad
        && \Gamma; \emptyset \vdash P :: y:A
        \tag{inversion on 1} \\
        & && 3\quad
        && P \vdcll \dual{(\Gamma)}; y:A
        \tag{IH on 2} \\
        & && 4\quad
        && \serv{x}{y}.P \vdcll \dual{(\Gamma)}; y:\bang A
        \tag{\ruleLabel{$\bang$} on 3}
        \displaybreak[1] \\[5pt]
        & \bullet \ruleLabel{$\bang$L}
        && 1\quad
        && \Gamma; \Delta, x:\bang A \vdash P\subst{x/u} :: \Lambda
        \tag{assumption} \\
        & && 2\quad
        && \Gamma, u:A; \Delta \vdash P :: \Lambda
        \tag{inversion on 1} \\
        & && 3\quad
        && P \vdcll \dual{(\Gamma)}, u:\dual{A}; \dual{(\Delta)}, \Lambda
        \tag{IH on 2} \\
        & && 4\quad
        && P\subst{x/u} \vdcll \dual{(\Gamma)}; \dual{(\Delta)}, \Lambda, x:\whynot \dual{A}
        \tag{\ruleLabel{$\whynot$} on 3}
        \displaybreak[1] \\[5pt]
        & \bullet \ruleLabel{$\whynot$R}
        && 1\quad
        && \Gamma; \Delta \vdash P\subst{x/u} :: \Lambda, x:\whynot \dual{A}
        \tag{assumption} \\
        & && 2\quad
        && \Gamma, u:A; \Delta \vdash P :: \Lambda
        \tag{inversion on 1} \\
        & && 3\quad
        && P \vdcll \dual{(\Gamma)}, u:\dual{A}; \dual{(\Delta)}, \Lambda
        \tag{IH on 2} \\
        & && 4\quad
        && P\subst{x/u} \vdcll \dual{(\Gamma)}; \dual{(\Delta)}, \Lambda, x:\whynot \dual{A}
        \tag{\ruleLabel{$\whynot$} on 3}
        \displaybreak[1] \\[5pt]
        & \bullet \ruleLabel{$\whynot$L}
        && 1\quad
        && \Gamma; x:\whynot A \vdash \serv{x}{y}.P :: \emptyset
        \tag{assumption} \\
        & && 2\quad
        && \Gamma; y:A \vdash P :: \emptyset
        \tag{inversion on 1} \\
        & && 3\quad
        && P \vdcll \dual{(\Gamma)}; y:\dual{A}
        \tag{IH on 2} \\
        & && 4\quad
        && \serv{x}{y}.P \vdcll \dual{(\Gamma)}; x:\bang \dual{A}
        \tag{\ruleLabel{$\bang$} on 3}
        \displaybreak[1] \\[5pt]
        & \bullet \ruleLabel{cutRL}
        && 1\quad
        && \Gamma; \Delta, \Delta' \vdash \nu{x}(P \| Q) :: \Lambda, \Lambda'
        \tag{assumption} \\
        & && 2\quad
        && \Gamma; \Delta \vdash P :: \Lambda, x:A \\
        & && 3\quad
        && \Gamma; \Delta', x:A \vdash Q :: \Lambda'
        \tag{inversion on 1} \\
        & && 4\quad
        && P \vdcll \dual{(\Gamma)}; \dual{(\Delta)}, \Lambda, x:A
        \tag{IH on 2} \\
        & && 5\quad
        && Q \vdcll \dual{(\Gamma)}; \dual{(\Delta')}, \Lambda', x:\dual{A}
        \tag{IH on 3} \\
        & && 6\quad
        && \nu{x}(P \| Q) \vdcll \dual{(\Gamma)}; \dual{(\Delta)}, \dual{(\Delta')}, \Lambda, \Lambda'
        \tag{\ruleLabel{cut} on 4 and 5}
        \displaybreak[1] \\[5pt]
        & \bullet \ruleLabel{cutLR}
        && 1\quad
        && \Gamma; \Delta, \Delta' \vdash \nu{x}(P \| Q) :: \Lambda, \Lambda'
        \tag{assumption} \\
        & && 2\quad
        && \Gamma; \Delta, x:A \vdash P :: \Lambda \\
        & && 3\quad
        && \Gamma; \Delta' \vdash Q :: \Lambda', x:A
        \tag{inversion on 1} \\
        & && 4\quad
        && P \vdcll \dual{(\Gamma)}; \dual{(\Delta)}, \Lambda, x:\dual{A}
        \tag{IH on 2} \\
        & && 5\quad
        && Q \vdcll \dual{(\Gamma)}; \dual{(\Delta')}, \Lambda', x:A
        \tag{IH on 3} \\
        & && 6\quad
        && \nu{x}(P \| Q) \vdcll \dual{(\Gamma)}; \dual{(\Delta)}, \dual{(\Delta')}, \Lambda, \Lambda'
        \tag{\ruleLabel{cut} on 4 and 5}
        \displaybreak[1] \\[5pt]
        & \bullet \ruleLabel{cutRR}
        && 1\quad
        && \Gamma; \Delta, \Delta' \vdash \nu{x}(P \| Q) :: \Lambda, \Lambda'
        \tag{assumption} \\
        & && 2\quad
        && \Gamma; \Delta \vdash P :: \Lambda, x:A \\
        & && 3\quad
        && \Gamma; \Delta' \vdash Q :: \Lambda', x:\dual{A}
        \tag{inversion on 1} \\
        & && 4\quad
        && P \vdcll \dual{(\Gamma)}; \dual{(\Delta)}, \Lambda, x:A
        \tag{IH on 2} \\
        & && 5\quad
        && Q \vdcll \dual{(\Gamma)}; \dual{(\Delta')}, \Lambda', x:\dual{A}
        \tag{IH on 3} \\
        & && 6\quad
        && \nu{x}(P \| Q) \vdcll \dual{(\Gamma)}; \dual{(\Delta)}, \dual{(\Delta')}, \Lambda, \Lambda'
        \tag{\ruleLabel{cut} on 4 and 5}
        \displaybreak[1] \\[5pt]
        & \bullet \ruleLabel{cutLL}
        && 1\quad
        && \Gamma; \Delta, \Delta' \vdash \nu{x}(P \| Q) :: \Lambda, \Lambda'
        \tag{assumption} \\
        & && 2\quad
        && \Gamma; \Delta, x:A \vdash P :: \Lambda \\
        & && 3\quad
        && \Gamma; \Delta', x:\dual{A} \vdash Q :: \Lambda'
        \tag{inversion on 1} \\
        & && 4\quad
        && P \vdcll \dual{(\Gamma)}; \dual{(\Delta)}, \Lambda, x:\dual{A}
        \tag{IH on 2} \\
        & && 5\quad
        && Q \vdcll \dual{(\Gamma)}; \dual{(\Delta')}, \Lambda', x:A
        \tag{IH on 3} \\
        & && 6\quad
        && \nu{x}(P \| Q) \vdcll \dual{(\Gamma)}; \dual{(\Delta)}, \dual{(\Delta')}, \Lambda, \Lambda'
        \tag{\ruleLabel{cut} on 4 and 5}
        \displaybreak[1] \\[5pt]
        & \bullet \ruleLabel{cut$\bang$R}
        && 1\quad
        && \Gamma; \Delta \vdash \nu{u}(P \| \serv{u}{x}.Q) :: \Lambda
        \tag{assumption} \\
        & && 2\quad
        && \Gamma, u:A; \Delta \vdash P :: \Lambda \\
        & && 3\quad
        && \Gamma; \emptyset \vdash Q :: x:A
        \tag{inversion on 1} \\
        & && 4\quad
        && P \vdcll \dual{(\Gamma)}, u:\dual{A}; \dual{(\Delta)}, \Lambda
        \tag{IH on 2} \\
        & && 5\quad
        && Q \vdcll \dual{(\Gamma)}; x:A
        \tag{IH on 3} \\
        & && 6\quad
        && \nu{u} ( P \| \serv{u}{x}.Q ) \vdcll \dual{(\Gamma)}; \dual{(\Delta)}, \Lambda
        \tag{\ruleLabel{cut$\whynot$R} on 4 and 5}
        \displaybreak[1] \\[5pt]
        & \bullet \ruleLabel{cut$\bang$L}
        && 1\quad
        && \Gamma; \Delta \vdash \nu{u}(\serv{u}{x}.P \| Q) :: \Lambda
        \tag{assumption} \\
        & && 2\quad
        && \Gamma; \emptyset \vdash P :: x:A \\
        & && 3\quad
        && \Gamma, u:A; \Delta \vdash Q :: \Lambda
        \tag{inversion on 1} \\
        & && 4\quad
        && P \vdcll \dual{(\Gamma)}; x:A
        \tag{IH on 2} \\
        & && 5\quad
        && Q \vdcll \dual{(\Gamma)}, u:\dual{A}; \dual{(\Delta)}, \Lambda
        \tag{IH on 3} \\
        & && 6\quad
        && \nu{u}(\serv{u}{x}.P \| Q) \vdcll \dual{(\Gamma)}; \dual{(\Delta)}, \Lambda
        \tag{\ruleLabel{cut$\whynot$L} on 4 and 5}
        \displaybreak[1] \\[5pt]
        & \bullet \ruleLabel{cut$\whynot$R}
        && 1\quad
        && \Gamma; \Delta \vdash \nu{u} ( P \| \serv{u}{x}.Q ) :: \Lambda
        \tag{assumption} \\
        & && 2\quad
        && \Gamma, u:A; \Delta \vdash P :: \Lambda \\
        & && 3\quad
        && \Gamma; x:\dual{A} \vdash Q :: \emptyset
        \tag{inversion on 1} \\
        & && 4\quad
        && P \vdcll \dual{(\Gamma)}, u:\dual{A}; \dual{(\Delta)}, \Lambda
        \tag{IH on 2} \\
        & && 5\quad
        && Q \vdcll \dual{(\Gamma)}; x:A
        \tag{IH on 3} \\
        & && 6\quad
        && \nu{u} ( P \| \serv{u}{x}.Q ) \vdcll \dual{(\Gamma)}; \dual{(\Delta)}, \Lambda
        \tag{\ruleLabel{cut$\whynot$R} on 4 and 5}
        \displaybreak[1] \\[5pt]
        & \bullet \ruleLabel{cut$\whynot$L}
        && 1\quad
        && \Gamma; \Delta \vdash \nu{u}(\serv{u}{x}.P \| Q) :: \Lambda
        \tag{assumption} \\
        & && 2\quad
        && \Gamma; x:\dual{A} \vdash P :: \emptyset \\
        & && 3\quad
        && \Gamma, u:A; \Delta \vdash Q :: \Lambda
        \tag{inversion on 1} \\
        & && 4\quad
        && P \vdcll \dual{(\Gamma)}; x:A
        \tag{IH on 2} \\
        & && 5\quad
        && Q \vdcll \dual{(\Gamma)}, u:\dual{A}; \dual{(\Delta)}, \Lambda
        \tag{IH on 3} \\
        & && 6\quad
        && \nu{u}(\serv{u}{x}.P \| Q) \vdcll \dual{(\Gamma)}; \dual{(\Delta)}, \Lambda
        \tag{\ruleLabel{cut$\whynot$L} on 4 and 5}
    \end{align*}

    \emph{($\,\C \subseteq \U$)}
    Take any $P \in \C$.
    Then, there are $\Gamma$, $\Delta$ s.t.\ $P \vdcll \Gamma; \Delta$.
    By showing that this implies $\dual{(\Gamma)}; \emptyset \vdash P :: \Delta$, we have $P \in \U$.
    As per the second item of \Cref{thm:madmissible}, Rules~\ruleLabel{$\mleft$} and~\ruleLabel{$\mright$} from \Cref{def:pullm} are admissible in \pull.
    Therefore, it suffices to show $\dual{(\Gamma)}; \emptyset \vpull P :: \Delta$.
    We do so by induction on the structure of the proof of $P \vdcll \Gamma; \Delta$.

    \begin{align*}
        & \bullet \ruleLabel{id}
        && 1\quad
        && \fwd{x}{y} \vdcll \Gamma; x:A, y:\dual{A}
        \tag{assumption} \\
        & && 2\quad
        && \dual{(\Gamma)}; x:\dual{A} \vpull \fwd{x}{y} :: y:\dual{A}
        \tag{\ruleLabel{idR}} \\
        & && 3\quad
        && \dual{(\Gamma)} ; \emptyset \vpull \fwd{x}{y} :: x:A , y:\dual{A}
        \tag{\ruleLabel{$\mright$}}
        \displaybreak[1] \\[5pt]
        & \bullet \ruleLabel{$\1$}
        && 1\quad
        && \send{x}{}.\0 \vdcll \Gamma; x:\1
        \tag{assumption} \\
        & && 2\quad
        && \dual{(\Gamma)}; \emptyset \vpull \send{x}{}.\0 :: x:\1
        \tag{\ruleLabel{$\1$R}}
        \displaybreak[1] \\[5pt]
        & \bullet \ruleLabel{$\bot$}
        && 1\quad
        && \recv{x}{}.P \vdcll \Gamma; \Delta, x:\bot
        \tag{assumption} \\
        & && 2\quad
        && P \vdcll \Gamma; \Delta
        \tag{inversion 1} \\
        & && 3\quad
        && \dual{(\Gamma)}; \emptyset \vpull P :: \Delta
        \tag{IH on 2} \\
        & && 4\quad
        && \dual{(\Gamma)}; x:\1 \vpull \recv{x}{}.P :: \Delta
        \tag{\ruleLabel{$\1$L} on 3} \\
        & && 5\quad
        && \dual{(\Gamma)}; \emptyset \vpull \recv{x}{}.P :: \Delta, x:\bot
        \tag{\ruleLabel{$\mright$} on 4}
        \displaybreak[1] \\[5pt]
        & \bullet \ruleLabel{$\tensor$}
        && 1\quad
        && \nu{y}\send{x}{y}.(P \| Q) \vdcll \Gamma; \Delta, \Delta', x:A \tensor B
        \tag{assumption} \\
        & && 2\quad
        && P \vdcll \Gamma; \Delta, y:A \\
        & && 3\quad
        && Q \vdcll \Gamma; \Delta', x:B
        \tag{inversion on 1} \\
        & && 4\quad
        && \dual{(\Gamma)}; \emptyset \vpull P :: \Delta, y:A
        \tag{IH on 2} \\
        & && 5\quad
        && \dual{(\Gamma)}; \emptyset \vpull Q :: \Delta', x:B
        \tag{IH on 3} \\
        & && 6\quad
        && \dual{(\Gamma)}; \emptyset \vpull \nu{y}\send{x}{y}(P \| Q) :: \Delta, \Delta', x:A \tensor B
        \tag{\ruleLabel{$\tensor$R} on 4 and 5}
        \displaybreak[1] \\[5pt]
        & \bullet \ruleLabel{$\parr$}
        && 1\quad
        && \recv{x}{y}.P \vdcll \Gamma; \Delta, x:A \parr B
        \tag{assumption} \\
        & && 2\quad
        && P \vdcll \Gamma; \Delta, y:A, x:B
        \tag{inversion on 1} \\
        & && 3\quad
        && \dual{(\Gamma)}; \emptyset \vpull P :: \Delta, y:A, x:B
        \tag{IH on 2} \\
        & && 4\quad
        && \dual{(\Gamma)}; y:\dual{A} \vpull P :: \Delta, x:B
        \tag{\ruleLabel{$\mleft$} on 3} \\
        & && 5\quad
        && \dual{(\Gamma)}; \emptyset \vpull \recv{x}{y}.P :: \Delta, x:\underbrace{\dual{A} \lolli B}_{A \parr B}
        \tag{\ruleLabel{$\lolli$R} on 4}
        \displaybreak[1] \\[5pt]
        & \bullet \ruleLabel{$\oplus$}
        && 1\quad
        && x \triangleleft j.P \vdcll \Gamma; \Delta, x:\oplus\{i:A_i\}_{i \in I}
        \tag{assumption} \\
        & && 2\quad
        && P \vdcll \Gamma; \Delta, x:A_j \\
        & && 3\quad
        && j \in I
        \tag{inversion on 1} \\
        & && 4\quad
        && \dual{(\Gamma)}; \emptyset \vpull P :: \Delta, x:A_j
        \tag{IH on 2} \\
        & && 5\quad
        && \dual{(\Gamma)}; \emptyset \vpull x \triangleleft j.P :: \Delta, x:\oplus\{i:A_i\}_{i \in I}
        \tag{\ruleLabel{$\oplus$R} on 4 and 3}
        \displaybreak[1] \\[5pt]
        & \bullet \ruleLabel{$\&$}
        && 1\quad
        && x \triangleright \{i:P_i\}_{i \in I} \vdcll \Gamma; \Delta, x:\&\{i:A_i\}_{i \in I}
        \tag{assumption} \\
        & && 2\quad
        && \forall i \in I.~ P_i \vdcll \Gamma; \Delta, x:A_i
        \tag{inversion on 1} \\
        & && 3\quad
        && \forall i \in I.~ \dual{(\Gamma)}; \emptyset \vpull P_i :: \Delta, x:A_i
        \tag{IH on 2} \\
        & && 4\quad
        && \dual{(\Gamma)}; \emptyset \vpull x \triangleright \{i:P_i\}_{i \in I} :: \Delta, x:\&\{i:A_i\}_{i \in I}
        \tag{\ruleLabel{$\&$R} on 3}
        \displaybreak[1] \\[5pt]
        & \bullet \ruleLabel{copy}
        && 1\quad
        && \nu{x}\send{u}{x}.P \vdcll \Gamma, u:A; \Delta
        \tag{assumption} \\
        & && 2\quad
        && P \vdcll \Gamma, u:A; \Delta, x:A
        \tag{inversion on 1} \\
        & && 3\quad
        && \dual{(\Gamma)}, u:\dual{A}; \emptyset \vpull P :: \Delta, x:A
        \tag{IH on 2} \\
        & && 4\quad
        && \dual{(\Gamma)}, u:\dual{A}; x:\dual{A} \vpull P :: \Delta
        \tag{\ruleLabel{$\mleft$} on 3} \\
        & && 5\quad
        && \dual{(\Gamma)}, u:\dual{A}; \emptyset \vpull \nu{x}\send{u}{x}.P :: \Delta
        \tag{\ruleLabel{copyL} on 4}
        \displaybreak[1] \\[5pt]
        & \bullet \ruleLabel{$\bang$}
        && 1\quad
        && \serv{x}{y}.P \vdcll \Gamma; x:\bang A
        \tag{assumption} \\
        & && 2\quad
        && P \vdcll \Gamma; y:A
        \tag{inversion on 1} \\
        & && 3\quad
        && \dual{(\Gamma)}; \emptyset \vpull P :: y:A
        \tag{IH on 2} \\
        & && 4\quad
        && \dual{(\Gamma)}; \emptyset \vpull \serv{x}{y}.P :: x:\bang A
        \tag{\ruleLabel{$\bang$R} on 3}
        \displaybreak[1] \\[5pt]
        & \bullet \ruleLabel{$\whynot$}
        && 1\quad
        && P\subst{x/u} \vdcll \Gamma; \Delta, x:\whynot A
        \tag{assumption} \\
        & && 2\quad
        && P \vdcll \Gamma, u:A; \Delta
        \tag{inversion on 1} \\
        & && 3\quad
        && \dual{(\Gamma)}, u:\dual{A}; \emptyset \vpull P :: \Delta
        \tag{IH on 2} \\
        & && 4\quad
        && \dual{(\Gamma)}; x:\bang \dual{A} \vpull P\subst{x/u} :: \Delta
        \tag{\ruleLabel{$\bang$L} on 3} \\
        & && 5\quad
        && \dual{(\Gamma)}; \emptyset \vpull P\subst{x/u} :: \Delta, x:\whynot A
        \tag{\ruleLabel{$\mright$} on 4}
        \displaybreak[1] \\[5pt]
        & \bullet \ruleLabel{cut}
        && 1\quad
        && \nu{x}(P \| Q) \vdcll \Gamma; \Delta, \Delta'
        \tag{assumption} \\
        & && 2\quad
        && P \vdcll \Gamma; \Delta, x:A \\
        & && 3\quad
        && Q \vdcll \Gamma; \Delta', x:\dual{A}
        \tag{inversion on 1} \\
        & && 4\quad
        && \dual{(\Gamma)}; \emptyset \vpull P :: \Delta, x:A
        \tag{IH on 2} \\
        & && 5\quad
        && \dual{(\Gamma)}; \emptyset \vpull Q :: \Delta', x:\dual{A}
        \tag{IH on 3} \\
        & && 6\quad
        && \dual{(\Gamma)}; x:A \vpull Q :: \Delta'
        \tag{\ruleLabel{$\mleft$} on 5} \\
        & && 7\quad
        && \dual{(\Gamma)}; \emptyset \vpull \nu{x}(P \| Q) :: \Delta, \Delta'
        \tag{\ruleLabel{cutRL} on 4 and 6}
        \displaybreak[1] \\[5pt]
        & \bullet \ruleLabel{cut$\whynot$R}
        && 1\quad
        && \nu{u} ( P \| \serv{u}{x}.Q ) \vdcll \Gamma; \Delta
        \tag{assumption} \\
        & && 2\quad
        && P \vdcll \Gamma, u:A; \Delta \\
        & && 3\quad
        && Q \vdcll \Gamma; x:\dual{A}
        \tag{inversion on 1} \\
        & && 4\quad
        && \dual{(\Gamma)}, u:\dual{A}; \emptyset \vpull P :: \Delta
        \tag{IH on 2} \\
        & && 5\quad
        && \dual{(\Gamma)}; \emptyset \vpull Q :: x:\dual{A}
        \tag{IH on 3} \\
        & && 6\quad
        && \dual{(\Gamma)}; \emptyset \vpull \nu{u} ( P \| \serv{u}{x}.Q ) :: \Delta
        \tag{\ruleLabel{cut$\bang$R}}
        \displaybreak[1] \\[5pt]
        & \bullet \ruleLabel{cut$\whynot$L}
        && 1\quad
        && \nu{u}(\serv{u}{x}.P \| Q) \vdcll \Gamma; \Delta
        \tag{assumption} \\
        & && 2\quad
        && P \vdcll \Gamma; x:\dual{A} \\
        & && 3\quad
        && Q \vdcll \Gamma, u:A; \Delta
        \tag{inversion on 1} \\
        & && 4\quad
        && \dual{(\Gamma)}; \emptyset \vpull P :: x:\dual{A}
        \tag{IH on 2} \\
        & && 5\quad
        && \dual{(\Gamma)}, u:\dual{A}; \emptyset \vpull Q :: \Delta
        \tag{IH on 3} \\
        & && 6\quad
        && \dual{(\Gamma)}; \emptyset \vpull \nu{u}(\serv{u}{x}.P \| Q) :: \Delta
        \tag{\ruleLabel{cut$\bang$L}}
    \end{align*}

    \noindent
    This concludes the proof of \Cref{thm:cllull}.
\end{proof}

\newpage
We now turn our attention to \I, the class of processes typable under the intuitionistic interpretation.
The observation by Caires, Pfenning and Toninho in~\cite{journal/mscs/CairesPT16} entails that \I (\pill-typable processes) should be a strict subset of \C (\pcllcp-typable processes).
We formalize this fact by characterizing the intuitionistic fragment of \pull (which coincides with \pill) by limiting its typing rules and by showing that the class of processes typable in this fragment coincides with \I (\Cref{thm:illullsull}).
It then follows that this class of processes is strictly contained in \U (\pull-typable processes; \Cref{thm:illsubsetull}).
Since we have just shown that $\U = \C$, this formalizes the fact that \I is a strict subset of \C.

\Cref{thm:illullsull} below formalizes two equivalent characterizations of the intuitionistic fragment of \pull.
One characterization is based on the limited form of duality of \pill:
the Rules~\ruleLabel{$\mleft$} and~\ruleLabel{$\mright$} in \Cref{def:pullm} may not be used.
The other characterization is based on the restricted two-sided form of \pill sequents:
the rules in \Cref{f:ull-inf} are limited to have exactly one channel/type pair on the right.
This requires the following auxiliary definition:

\begin{definition}[\rdegree]
    The \emph{\rdegree} of a \pull sequent $\Gamma; \Delta \vdash P :: \Lambda$ is the size of $\Lambda$.
    We say this sequent has \rdegree $|\Lambda|$.
\end{definition}

We now have:

\begin{theorem}\label{thm:illullsull}
    Given a process $P$, the following are equivalent:
    \begin{enumerate}
        \item there are $\Gamma$, $\Delta$, $x$ and $A$ such that $\Gamma; \Delta \vdill P :: x:A$;
        \item there are $\Gamma$, $\Delta$ and $\Lambda$ such that $\Gamma; \Delta \vdash P :: \Lambda$ where all sequents in its proof have \rdegree 1;
        \item there are $\Gamma$, $\Delta$ and $\Lambda$ such that $\Gamma; \Delta \vpull P :: \Lambda$ where its proof never uses Rules~\ruleLabel{$\mleft$} and~\ruleLabel{$\mright$}.
    \end{enumerate}
\end{theorem}

\begin{proof}
We first argue that items 1 and 2 are equivalent; then, we argue that items 2 and 3 imply each other:

\smallskip

    \emph{(Equivalence of items 1 and 2)}
    We show that restricting the sequents in \pull to an \rdegree of 1 entails the limitation of its rules to a strict subset: those marked with $\ast$ in \Cref{f:ull-inf}.
    This set of rules coincides with the set of rules for \pill in \Cref{f:ill-inf}.
    Hence, a proof of typability in \pill (item~1) can be replicated in \pull with \rdegree 1 (item~2) and vice versa.
    Because the proof for each rule follows the same pattern, we only detail two cases: one without $\ast$ and one with $\ast$.
    \begin{itemize}
        \item
            Rule~\ruleLabel{$\parr$L} is not $\ast$-marked.
            Suppose we have a proof of typability in \pull with \rdegree 1 and suppose this includes an application of \ruleLabel{$\parr$L}.
            By assumption, this application's antecedents have \rdegree 1, so its consequence must have \rdegree 2.
            This contradicts the assumption that all sequents have \rdegree 1, showing that Rule~\ruleLabel{$\parr$L} is not usable in this fragment of \pull.

        \item
            Rule~\ruleLabel{$\lolli$L} is $\ast$-marked.
            By assumption, in any application of this rule its antecedents would have \rdegree 1.
            Then, its consequence also has \rdegree 1, so this rule can be used without problems.
    \end{itemize}

    \emph{(Item 2 implies item 3)}
    By \Cref{def:pullm}, the set of rules of \pullm consists of only \ruleLabel{$\mleft$}, \ruleLabel{$\mright$}, and the $\ast$-marked rules in \Cref{f:ull-inf}.
    In \pullm, the only rules that alter the \rdegree{s} in judgments are \ruleLabel{$\mleft$} and \ruleLabel{$\mright$}.
    Hence, if these rules may not be used, proofs of typability in \pullm have a constant \rdegree.
    Since all axioms of \pullm (the $\ast$-marked Axioms~\ruleLabel{$\id$R} and~\ruleLabel{$\1$R} in \Cref{f:ull-inf}) have \rdegree 1, this constant \rdegree is 1.
    Therefore, a proof of typability in \pullm without using \ruleLabel{$\mleft$} and \ruleLabel{$\mright$} (item 3) can be replicated exactly in \pull with \rdegree restricted to~1 (item~2).

    \emph{(Item 3 implies item 2)}
    The proof of equivalence of items 1 and 2 above shows that with \rdegree~1 the only usable rules in \pull are those marked with $\ast$ in \Cref{f:ull-inf}.
    By \Cref{def:pullm}, without Rules~\ruleLabel{$\mleft$} and~\ruleLabel{$\mright$}, \pullm consists only of these $\ast$-marked rules.
    Hence, a proof of typability in \pull with \rdegree restricted to 1 (item 2) can be replicated exactly in \pullm without using \ruleLabel{$\mleft$} and \ruleLabel{$\mright$} (item 3).
\end{proof}

We may now state our final result:

\begin{theorem}\label{thm:illsubsetull}
    $\I \subsetneq \U$.
\end{theorem}

\begin{proof}
    \label{proof:illsubsetull}
    \Cref{thm:illullsull} ensures that $\I \subseteq \U$.
    To show that this inclusion is strict, we give a process $P \in \U$ such that $P \notin \I$.
    Assuming a proof for $u:B; \emptyset \vdash P' :: z:A$, let $P = \recv{x}{y}.\serv{y}{z}.P' \subst{x/u}$;
    there are precisely three ways to prove $P \in \U$: they follow from the three ways to infer a receive in \pull (\ruleLabel{$\parr$R}, \ruleLabel{$\lolli$R}, and \ruleLabel{$\tensor$L}).
    The proofs are as follows:
    \begin{mathpar}
        \begin{bussproof}
            \bussAssume{
                u:B; \emptyset \vdash P' :: z:A
            }
            \bussUn[\ruleLabel{$\bang$R}]{
                u:B; \emptyset \vdash \serv{y}{z}.P' :: y:\bang A
            }
            \bussUn[\ruleLabel{$\whynot$R}]{
                \emptyset; \emptyset \vdash \serv{y}{z}.P' \subst{x/u} :: y:\bang A, x:\whynot \dual{B}
            }
            \bussUn[\ruleLabel{$\parr$R}]{
                \emptyset; \emptyset \vdash \recv{x}{y}.\serv{y}{z}.P' \subst{x/u} :: x:(\bang A) \parr (\whynot \dual{B})
            }
        \end{bussproof}
        \and
        \begin{bussproof}
            \bussAssume{
                u:B; \emptyset \vdash P' :: z:A
            }
            \bussUn[\ruleLabel{$\mleft$}]{
                u:B; z:\dual{A} \vdash P' :: \emptyset
            }
            \bussUn[\ruleLabel{$\whynot$L}]{
                u:B; y:\whynot \dual{A} \vdash \serv{y}{z}.P' :: \emptyset
            }
            \bussUn[\ruleLabel{$\whynot$R}]{
                \emptyset; y:\whynot \dual{A} \vdash \serv{y}{z}.P' \subst{x/u} :: x:\whynot \dual{B}
            }
            \bussUn[\ruleLabel{$\lolli$R}]{
                \emptyset; \emptyset \vdash \recv{x}{y}.\serv{y}{z}.P' \subst{x/u} :: x:(\whynot \dual{A}) \lolli (\whynot \dual{B})
            }
        \end{bussproof}
        \and
        \begin{bussproof}
            \bussAssume{
                u:B; \emptyset \vdash P' :: z:A
            }
            \bussUn[\ruleLabel{$\mleft$}]{
                u:B; z:\dual{A} \vdash P' :: \emptyset
            }
            \bussUn[\ruleLabel{$\whynot$L}]{
                u:B; y:\whynot \dual{A} \vdash \serv{y}{z}.P' :: \emptyset
            }
            \bussUn[\ruleLabel{$\bang$L}]{
                \emptyset; y:\whynot \dual{A}, x:\bang B \vdash \serv{y}{z}.P' \subst{x/u} :: \emptyset
            }
            \bussUn[\ruleLabel{$\tensor$L}]{
                \emptyset; x:(\whynot \dual{A}) \tensor (\bang B) \vdash \recv{x}{y}.\serv{y}{z}.P' \subst{x/u} :: \emptyset
            }
        \end{bussproof}
    \end{mathpar}
    Clearly, all these proofs contain judgments of \rdegree different from 1 (first case) or require using \ruleLabel{$\mleft$}/\ruleLabel{$\mright$} (second and third cases).
    Hence, by \Cref{thm:illullsull}, $P \notin \I$.
\end{proof}

\section{Analysis}
\label{sec:discussion}

Now that we have established characterizations for the intuitionistic and classical fragments of \pull, we analyze the meaning of these results and discuss possible extensions.

Our analysis is twofold.
First, we consider the informal observation by Caires, Pfenning, and Toninho~\cite{journal/mscs/CairesPT16} that intuitionistic type systems enforce the locality principle~(\secref{ssec:locality});
unlike the classical formulation, the intuitionistic fragment of \pull (and thus \pill) has a partial form of duality, which ensures that typability guarantees locality.
Next, we discuss how \pull can support alternative, more expressive forms of parallel composition and restriction, and how such extensions transfer to classical or intuitionistic type systems (\secref{ssec:mixcycle});
we will see that the complete duality of classical type systems allows for such extensions, while the rely-guarantee reading of intuitionistic type systems cannot account for them.

\subsection{Locality}
\label{ssec:locality}

Locality is a well-known principle in concurrency research~\cite{conf/fossacs/Merro00}.
The idea is that freshly created channels are \emph{local}.
Local channels are \emph{mutable}, in the sense that they can be used for receives.
Once a channel has been transmitted to another location, it becomes \emph{non-local}, and thus \emph{immutable}:
it can only be used for sends---receives are no longer allowed.
This makes locality particularly relevant for giving formal semantics to distributed programming languages;
a prime example is the \emph{join calculus}~\cite{journal/tcs/FournetL01}, whose theory relies on (and is deeply influenced by) the locality principle~\cite{conf/concur/FournetGLMR96}.

Locality also makes an appearance in other contexts.
Honda and Laurent's~\cite{journal/tcs/HondaL10} correspondence between a typed $\pi$-calculus and polarized proof-nets enforces locality due to receives having negative polarity while received names may only be of positive polarity.
Also, Dal Lago \etal's~\cite{conf/popl/DalLagoVMY19} typed $\pi$-calculus with intersection types goes further:
processes are \emph{hyperlocalized}, meaning that there may be no receives on free channels after a prior receive at all.

Neither the intuitionistic nor the classical interpretation guarantee full-fledged locality through typing:
both systems allow receives on previously received channels.
The exception is \emph{replicated receive}, which is used to define a shared channel that can continuously receive linear channels over which to perform a service.
The intuitionistic interpretation guarantees locality for shared channels;
in other words, \pill cannot type replicated receives on non-local channels.

The following example, taken from the work by Caires, Pfenning and Toninho~\cite{journal/mscs/CairesPT16}, is typable in \pcllcp but not in \pill:
\[
    \nu{x}(\recv{x}{y}.\serv{y}{z}.P \| \nu{q}\send{x}{q}.Q)
\]

\noindent
Consider the left process in the parallel composition, $\recv{x}{y}.\serv{y}{z}.P$.
It first receives a channel $y$ over channel $x$; then, it uses $y$ for replicated receive, thus defining it as a shared channel.
In \pcllcp, channel $x$ has type $\bang A \parr B$.
The intuitionistic variant of this type is $\dual{(\bang A)} \lolli B = \whynot \dual{A} \lolli B$ on the right of the typing judgment, and $\whynot \dual{A} \tensor \dual{B}$ on the left.
It is impossible to type a process with a channel of such a type in \pill, \beginjp because `\whynot' is not a connective in the intuitionistic system. \endjp

The fact that \pill cannot type non-local shared channels---due to the absence of a dual for~`\bang'---suggests that there should be another kind of non-local channels that \pill cannot type:
there is no dual for \1.
Indeed, \pill cannot type empty sends on previously received channels,  a feature that is not observed in~\cite{journal/mscs/CairesPT16}.
The following example is typable in \pcllcp but not in \pill:
\[
    \recv{x}{y}.\recv{x}{}.\send{y}{}.\0
\]
In \pcllcp, the type of $x$ in this process is $\1 \parr \bot$.
The intuitionistic variant of this type is $\dual{\1} \lolli \bot = \bot \lolli \bot$ on the right of the typing judgment, and $\bot \tensor \1$ on the left.
This process is not typable in \pill, because \beginjp $\bot$ is not an intuitionistic proposition and so the type system has no rules to type \bot channels. \endjp
We make this observation more precise by giving an alternative proof to \Cref{thm:illsubsetull} ($\I \subsetneq \U$) based on this observation:

\begin{proof}[Alternative proof of \Cref{thm:illsubsetull}]
    By \Cref{thm:illullsull}, it is sufficient to give $P \in \U$ such that $P \notin \I$.
    Let $P := \recv{x}{y} . \recv{x}{} . \send{y}{} . \0$.
    There are three ways to prove that $P \in \U$:
    \begin{mathpar}
        \begin{bussproof}
            \bussAx[\ruleLabel{$\bot$L}]{
                \Gamma ; y:\bot \vdash \send{y}{} . \0 :: \emptyset
            }
            \bussUn[\ruleLabel{$\1$L}]{
                \Gamma ; y:\bot , x:\1 \vdash \recv{x}{} . \send{y}{} . \0 :: \emptyset
            }
            \bussUn[\ruleLabel{$\tensor$L}]{
                \Gamma ; x:\bot \tensor \1 \vdash \recv{x}{y} . \recv{x}{} . \send{y}{} . \0 :: \emptyset
            }
        \end{bussproof}
        \and
        \begin{bussproof}
            \bussAx[\ruleLabel{$\1$R}]{
                \Gamma ; \emptyset \vdash \send{y}{} . \0 :: y:\1
            }
            \bussUn[\ruleLabel{$\bot$R}]{
                \Gamma ; \emptyset \vdash \recv{x}{} . \send{y}{} . \0 :: y:\1 , x:\bot
            }
            \bussUn[\ruleLabel{$\parr$R}]{
                \Gamma ; \emptyset \vdash \recv{x}{y} . \recv{x}{} . \send{y}{} . \0 :: x:\1 \parr x:\bot
            }
        \end{bussproof}
        \and
        \begin{bussproof}
            \bussAx[\ruleLabel{$\bot$L}]{
                \Gamma ; y:\bot \vdash \send{y}{} . \0 :: \emptyset
            }
            \bussUn[\ruleLabel{$\bot$R}]{
                \Gamma ; y:\bot \vdash \recv{x}{} . \send{y}{} . \0 :: x:\bot
            }
            \bussUn[\ruleLabel{$\lolli$Rl}]{
                \Gamma ; \emptyset \vdash \recv{x}{y} . \recv{x}{} . \send{y}{} . \0 :: x:\bot \lolli \bot
            }
        \end{bussproof}
    \end{mathpar}
    Clearly, all these proofs contain judgments of \rdegree different from 1.
    Hence, by \Cref{thm:illullsull}, $P \notin \I$.
\end{proof}

\plscheck{Based on these two proofs for  \Cref{thm:illsubsetull}, we have corroborated Caires \etal's observation on locality, and extended it to the case of empty sends and receives.}

\paragraph{Judgments, Revisited}
\beginjp As already discussed, `$\bot$' and `$\whynot$' are not propositions in \pill;
the absence of rules for  them in \pill is  an inherent consequence of the form of its judgments. \endjp
As shown by \Cref{thm:illullsull}, rules for `$\whynot$' and `$\bot$' are an impossibility when judgments are required to have exactly one assignment (channel/proposition) on the right. 
\Cref{thm:illullsull} also shows this is closely related to duality, which is not \beginbas total \endbas in the intuitionistic interpretation \beginjp (i.e., the duals for `$\1$' and `$!$' given by \Cref{def:duality} are not propositions in \pill). \endjp
In the classical case, the type system is \emph{symmetrical} (as shown by a support for \ruleLabel{$\mleft$}/\ruleLabel{$\mright$}-rules), while the intuitionistic type system is \emph{asymmetrical}.

\beginjp
Concerning `$\bot$', the judgmental presentation of intuitionistic linear logic in~\cite{report/ChangCP03} defines an alternative approach, which allows deriving `$\bot$' as `multiplicative contradiction' by using a judgmental account of possibility under which resources are used zero times. \endjp

\subsection{Parallel Composition and Restriction}
\label{ssec:mixcycle}

The type systems we have discussed so far are rather restrictive in terms of how processes can be composed and connected:
there is only the cut-rule which composes and connects two processes that have exactly one channel of dual type in common.
Indeed, the cut-rule jointly handles constructs for parallel composition and restriction, which contrasts with many non-logical type systems for the $\pi$-calculus (e.g., \cite{journal/toplas/KobayashiPT99,conf/unu/Kobayashi03,journal/ic/Vasconcelos12}) where parallel composition and restriction typically have a dedicated rule each.

In this section we discuss extensions to the type systems that decompose cut into two separate rules:
mix for parallel composition, and cycle for channel connection (restriction).
As we will see, these notions form another clear distinction point between Curry-Howard interpretations of classical and intuitionistic linear logic.

\subsubsection{Independent Parallel Composition}
\label{ss:silent}
Caires \etal~\cite{journal/mscs/CairesPT16} discuss an alternative form of parallel composition. In the intuitionistic interpretation, the so-called \emph{independent} parallel composition
connects (i)~an arbitrary process $Q$ with (ii)~a process $P$ with a channel of type \1 on the right; it is derivable in the \emph{silent} interpretation of \1, in which
 Axioms~\ruleLabel{$\1$R} and~\ruleLabel{$\bot$L}   type the inactive process $\0$ and Rules~\ruleLabel{$\bot$R} and~\ruleLabel{$\1$L} leave processes untouched.
This way, the rules for $\1$ would be the following (the rules for $\bot$ are similar):
\begin{mathpar}
     \begin{bussproof}[$\1$R]
        \bussAx{
            \Gamma; \emptyset \vdash \0 :: x:\1
        }
    \end{bussproof}
    \and
    \begin{bussproof}[$\1$L]
        \bussAssume{
            \Gamma; \Delta \vdash P :: \Lambda
        }
        \bussUn{
            \Gamma; \Delta, x:\1 \vdash P :: \Lambda
        }
    \end{bussproof}
\end{mathpar}

In this case, the correspondence relies on structural congruences in processes and proofs (suitably extended), 
rather than on reduction.
To obtain independent parallel composition, one uses Rule~\ruleLabel{cut} to connect the right channel of $P$ with a corresponding (dual) channel in $Q$, which is exposed on the left using Rule~\ruleLabel{$\1$L}:
\[
    \begin{bussproof}
        \bussAssume{
            \Gamma; \Delta \vdill P :: z:\1
        }
        \bussAssume{
            \Gamma; \Delta' \vdill Q :: x:C
        }
        \bussUn[\ruleLabel{$\1$L}]{
            \Gamma; \Delta', z:\1 \vdill Q :: x:C
        }
        \bussBin[\ruleLabel{cutRL}]{
            \Gamma; \Delta, \Delta' \vdill \nu{z}(P \| Q) :: x:C
        }
    \end{bussproof}
\]

The requirement that process $P$ above has a channel of type \1 on the right is rather restrictive:
only left-rules can be used for typing, so $P$ must be a process that relies on multiple behaviors but can offer a behavior of type~$\1$.
Related to this constraint,
Caires \etal show that in the silent interpretation
$\Gamma; \Delta \vdill P :: x:A$ implies
$\Gamma; \Delta,x:\dual{A} \vdill P :: z:\1$, for some fresh $z$, provided that $A$ is exponential-free \beginbas and $\dual{\1} = \1$ (cf.\ also below)\endbas~\cite[Prop.~5.1]{journal/mscs/CairesPT16}.
However, this ``movement'' from the right to the left is only possible  under the absence of the identity axiom, which is in turn necessary when considering, e.g., behavioral polymorphism~\cite{conf/esop/CairesPPT13}.
Indeed, given any type $A$ we can prove $\emptyset; x:A \vdill \fwd{x}{y} :: y:A$ using the identity axiom, but there is no way to prove $\emptyset; x:A, y:\dual{A} \vdill \fwd{x}{y} :: z:\1$.

\subsubsection{Parallel Composition: Mix}

Girard~\cite{journal/tcs/Girard87} discusses an extension of linear logic that allows the combination of two independent proofs.
Wadler gives a Curry-Howard interpretation of this rule as the parallel composition of two processes that have no channels in common~\cite{conf/icfp/Wadler12}.
In~\cite{report/Caires14}, Caires complements the extension with an axiom that types the inactive process \0 with an empty context.

We now define \pcllcpp, which is \pcllcp extended with this form of parallel composition.
The only change is the addition of the following rules:
\begin{mathpar}
    \begin{bussproof}[mix]
        \bussAssume{
            P \vdcll \Gamma; \Delta
        }
        \bussAssume{
            Q \vdcll \Gamma; \Delta'
        }
        \bussBin{
            P \| Q \vdcll \Gamma; \Delta, \Delta'
        }
    \end{bussproof}
    \and
    \begin{bussproof}[empty]
        \bussAx{
            \0 \vdcll \Gamma; \emptyset
        }
    \end{bussproof}
\end{mathpar}
We straightforwardly define \pullp as a similar extension of \pull by adding the following rules:
\begin{mathpar}
    \begin{bussproof}[mix]
        \bussAssume{
            \Gamma; \Delta \vdash P :: \Lambda
        }
        \bussAssume{
            \Gamma; \Delta' \vdash Q :: \Lambda'
        }
        \bussBin{
            \Gamma; \Delta, \Delta' \vdash P \| Q :: \Lambda, \Lambda'
        }
    \end{bussproof}
    \and
    \begin{bussproof}[empty]
        \bussAx{
            \Gamma; \emptyset \vdash \0 :: \emptyset
        }
    \end{bussproof}
\end{mathpar}
It is easy to see that \Cref{thm:cllull}---the equivalence of \pcllcp and \pull---still holds for \pcllcpp and \pullp.
Following the reasoning of \Cref{thm:illullsull}---the characterization of \pill in terms of \pull---, Rules~\ruleLabel{mix} and~\ruleLabel{empty} do not belong to the intuitionistic fragment of \pullp:
Rule~\ruleLabel{empty} has \rdegree 0, and if the antecedents of Rule~\ruleLabel{mix} both have \rdegree 1 then its consequence would have \rdegree 2.
This leads us to conclude that the more flexible ways of composition induced by~\ruleLabel{mix} cannot be supported by intuitionistic systems.

\paragraph{The Equivalence of \1 and \bot}

Caires~\cite{report/Caires14} notes that, in the classical setting, Rules~\ruleLabel{mix} and \ruleLabel{empty} make it possible to prove $\bot \lolli \1$ and $\1 \lolli \bot$ (where \lolli denotes linear implication).
Hence, it is possible to consider \1 and \bot equivalent, writing a single symbol (say, `$\bullet$') for either---very similar to the singular, self-dual type \tend in standard session types.
The work of Atkey \etal~\cite{chapter/AtkeyLM16} is also relevant, as it develops a detailed treatment of the conflation of $\bot$ and $\1$.
\beginjp
Conflation is not always desirable: the work of Qian \etal~\cite{conf/icfp/QianKB21} addresses the shortcomings of considering \ruleLabel{mix} for faithfully representing stateful client-server behaviors.
\endjp

\subsubsection{Channel Connection: Cycle}

The cut-rule only allows us to connect two processes on a single channel.
Although this neatly guarantees deadlock freedom, it prevents many deadlock-free processes from being typable (see~\cite{journal/jlamp/DardhaP22} for comparisons between these classes of processes).
For example, we can construct a process $\nu{x}\nu{y}(P \| Q)$ where
\begin{align*}
    P
    &:= \nu{u}\send{x}{u}.(\recv{y}{v}.(\recv{u}{}.\recv{v}{}.\0 \| \send{y}{}.\0) \| \send{x}{}.\0) ~\text{and} \\
    Q
    &:= \recv{x}{w}.\nu{z}\send{y}{z}.(\recv{x}{}.\send{z}{}.\0 \| \recv{y}{}.\send{w}{}.\0)
\end{align*}
This process is deadlock free, but not typable using cut.

Connecting on multiple channels at once does have the danger of introducing the \emph{circular dependencies} between sessions that are at the heart of deadlocked processes.
For example, suppose that we replace $P$ with $P'$, where we swap the send on $x$ and the receive on $y$:
\[
    P' := \recv{y}{v}.\nu{u}\send{x}{u}.(\recv{u}{}.\recv{v}{}.\0 \| \send{y}{}.\0 \| \send{x}{}.\0).
\]
Now, the composed and connected process would be stuck with $P'$ waiting for a receive on $y$ and $Q$ for a receive on $x$.

Dardha and Gay~\cite{conf/fossacs/DardhaG18} present a session type system based on classical linear logic in which they replace the cut-rule with a rule called \ruleLabel{cycle}, which connects two channels of dual types in the same judgment.
This new rule has a side-condition that allows their type system to maintain deadlock freedom.
We can define \pcllcppp as an extension of \pcllcpp with the cycle-rule.
Since we are specifically interested in this form of channel restriction, below we abstract away from Dardha and Gay's side-condition by simply writing $\phi$.

Note that we first require some modifications to restriction and reduction, based on~\cite{journal/ic/Vasconcelos12}.
Restriction in \pcllcppp should involve two channel names, i.e. writing $\nu{x y}P$ instead of $\nu{x}P$, because the involved channels appear in the same judgment and thus need to be uniquely named.
Indeed, $\nu{x y}P$ says that $x$ and $y$ are the two endpoints of the same channel in $P$.
\[
    \begin{bussproof}[cycle]
        \bussAssume{
            P \vdcll \Gamma; \Delta, x:A, y:\dual{A}
        }
        \bussAssume{
            \phi
        }
        \bussBin{
            \nu{x y}P \vdcll \Gamma; \Delta
        }
    \end{bussproof}
\]
Defining \pullpp as a similar extension of \pullp is again straightforward.
Similarly to the cut-rules in \Cref{f:ull-inf-cut}, we add four rules that operate on different sides of the typing judgment:
\begin{mathpar}
    \begin{bussproof}[cycleRL]
        \bussAssume{
            \Gamma; \Delta, x:A \vdash P :: \Lambda, y:A
        }
        \bussAssume{
            \phi
        }
        \bussBin{
            \Gamma; \Delta \vdash \nu{y x}P :: \Lambda
        }
    \end{bussproof}
    \and
    \begin{bussproof}[cycleLR]
        \bussAssume{
            \Gamma; \Delta, x:A \vdash P :: \Lambda, y:A
        }
        \bussAssume{
            \phi
        }
        \bussBin{
            \Gamma; \Delta \vdash \nu{x y}P :: \Lambda
        }
    \end{bussproof}
    \and
    \begin{bussproof}[cycleRR]
        \bussAssume{
            \Gamma; \Delta \vdash P :: \Lambda, x:A, y:\dual{A}
        }
        \bussAssume{
            \phi
        }
        \bussBin{
            \Gamma; \Delta \vdash \nu{x y}P :: \Lambda
        }
    \end{bussproof}
    \and
    \begin{bussproof}[cycleLL]
        \bussAssume{
            \Gamma; \Delta, x:A, y:\dual{A} \vdash P :: \Lambda
        }
        \bussAssume{
            \phi
        }
        \bussBin{
            \Gamma; \Delta \vdash \nu{x y}P :: \Lambda
        }
    \end{bussproof}
\end{mathpar}

Again, it is easy to see that \Cref{thm:cllull} still holds for \pcllcppp and \pullpp.
Interestingly, by virtue of \Cref{thm:illullsull}, Rule~\ruleLabel{cycleLL} is \beginjp the only rule allowed \endjp in the  intuitionistic fragment of \pullpp, \beginjp because it does not affect the right-hand side of the judgment (as $\Lambda$ can have \rdegree 1). \endjp
Unfortunately, this rule is useless for \pill{}\beginbas: it requires two channels of dual types, but duality is not total in \pill. \endbas

{
    Adding mix- and cycle-rules enables a more flexible formulation of structural congruence, as well as more flexible typing of sends.
    For one, e.g., Rule~\ruleLabel{$\tensor$R} would only require a single continuation that provides both the payload and the continuation of the send.
    We can then add structural congruence rules that are found in traditional, untyped settings, such as $P \| Q \equiv Q \| P$ and $P \| \0 \equiv P$.
}

\section{Conclusion}
\label{sec:conclusion}

Curry-Howard correspondences between linear logic and session types explain how
linear logic can provide an ample, principled framework to conceive different typing disciplines for message-passing processes specified in the $\pi$-calculus.
There is no single canonical interpretation, as there are multiple interpretations, depending on design choices involving, e.g., the logical connectives considered and their respective process operators.

In this context, this paper has pursued a very concrete goal: to formally compare the interpretations of classical and intuitionistic linear logic as session types.  The comparison results reported in \Cref{sec:comparison} are an indispensable step to consolidating the logical foundations of message-passing concurrency; they also have a number of relevant ramifications, as reported in \Cref{sec:discussion}.

Our technical approach to this goal relies on a fragment of Girard's Logic of Unity (LU)~\cite{journal/apal/Girard93} as a basis to develop  \pull: a new session type system that can type all processes typable in both \pcllcp and \pill, which correspond to classical and intuitionistic interpretations, respectively.
{The linear logic \ull, on which the type system \pull stands, is an admittedly modest fragment of Girard's LU. Still, we emphasize that \pull is sufficient for our purposes, as it gives us a fair and rigorous basis for addressing our declared goal of comparing type systems based on different linear logics.
The development of computational interpretations for LU, including but going beyond \ull and \pull, is surely an interesting  direction but one that lies outside the scope of our work.}


In \pill, judgments have a particular reading in which a process \emph{relies} on several channels (on the left of judgments) and \emph{guarantees} a behavior along a single designated channel (on the right).
\pull uses two-sided judgments as well, similar to LU.
However, \pull does not retain \pill's rely-guarantee reading, because it does not distinguish between the sides of its sequents.
The consequence is that \pull supports a full duality relation, as opposed to \pill.
This allows \pull to mimic the explicit duality in the single-sided \pcllcp.
On the other hand, restricting the right side of \pull's judgments to exactly one channel---thus limiting support for duality---characterizes a fragment of \pull's typing rules that precisely coincides with \pill.

Our results confirm the informal observation by Caires \etal~\cite{journal/mscs/CairesPT16} that the difference between session type systems based on classical and intuitionistic linear logic is in the enforcement of locality of shared names.
We have not only confirmed that \pill cannot type processes that do not respect locality for shared channels:
a new insight obtained in this paper is that \pill also forbids empty sends on received channels.
Our results  show that these constraints are a consequence of the strict requirement that \pill's typing judgments must have exactly one channel/type on the right.
\pull and \pcllcp do not have such constraints and thus fully support duality.

\plscheck{Our work can be seen as providing a  technical answer to long-suspected but fuzzily understood differences between intuitionistic and classical interpretations of linear logic (here represented by \pill and \pcllcp, respectively) concerning the role of duality (implicit in \pill, explicit in \pcllcp) and locality of shared channels (enforced by \pill but not by \pcllcp). We believe it is important to give a formal footing to this kind of informal observations. Formally stating the required comparisons is not obvious, and insightful in itself, as implicit assumptions may emerge in the process. In this respect, we adopted \pull as it provides an effective yardstick for comparisons, independent from both \pill and \pcllcp, and based on an already existing framework (Girard's LU). In future work, it would be interesting to explore other approaches towards our technical results (\Cref{thm:cllull,thm:illsubsetull}).}

If one adopts the stance that a \emph{permissive} type system is also an \emph{expressive} one, then the results from our comparisons (notably, the strict inclusion of \pill in \pull) indicate that classical linear logic induces a more expressive class of typed processes than intuitionistic linear logic.
There is an alternative stance, which considers \emph{\beginbas less \endbas permissive} type systems as being \emph{more precise} than other, more permissive type system.
The locality property for shared channels, as enforced by intuitionistic interpretations, is a case in point here.
In our view, the connection between permissiveness, expressiveness, and precision is an interesting question, which largely depends on the value of the intended properties.
Indeed, the precision of type systems derived from intuitionistic interpretations may not be meaningful in settings/applications where locality is simply not a relevant property.
On the other hand, the conditions that make intuitionistic interpretations less permissive can always be imposed on more permissive interpretations as side-conditions, making them more precise.
Finally, as observed by Caires \etal~\cite{journal/mscs/CairesPT16}, it is surely remarkable that intuitionistic interpretations of session types precisely capture a principle such as locality of shared names, which was known and exploited in different contexts~\cite{conf/fossacs/Merro00} long before these Curry-Howard interpretations were first spelled out~\cite{conf/concur/CairesP10}.

To conclude, we believe that LU and \pull have other useful applications besides the comparison of type systems based on vanilla linear logic;
our discussion of more expressive parallel composition and channel connection in \Cref{sec:discussion} corroborates this.
We have been able to rely on \pull to study the effects of mix- and cycle-rules in the intuitionistic interpretation by transferring extensions of the classical interpretation to \pull and following the methodology of our comparison results.
In this case, extensions of intuitionistic interpretations with such rules do not make sense due to the constraints of typing judgments.
This makes sense from a logical perspective:
intuitionistic argumentation is based on a constructivist philosophy in which putting unrelated things together (\ruleLabel{mix}) and then relating them later (\ruleLabel{cycle}) may not be acceptable.


\paragraph*{Acknowledgements}

We are grateful to
Juan C. Jaramillo,
Joseph Paulus,
and
Revantha Ramanayake
for helpful discussions.
We would also like to thank the anonymous reviewers of PLACES\textquotesingle20 and JLAMP for their detailed suggestions, which were helpful to improve the presentation.

\DeclareRobustCommand{\VAN}[3]{#3} 

\label{sec:bib}
\bibliographystyle{alpha}
\bibliography{refs}

\end{document}